%% file: quantumxor.tex
\documentclass[11pt]{article}
\usepackage{fullpage}
\usepackage{amsmath,amsfonts,amsthm,mathrsfs,mathpazo,xspace,hyperref,graphicx,bm}
\usepackage{needspace}
\usepackage{endnotes}
\usepackage{mathrsfs}

\newtheorem{theorem}{Theorem}[section]

\newtheorem{proposition}[theorem]{Proposition}
\newtheorem{lemma}[theorem]{Lemma}

\newtheorem{claim}[theorem]{Claim}
\newtheorem{fact}[theorem]{Fact}

\newtheorem{definition}[theorem]{Definition}

\theoremstyle{remark}
\newtheorem{example}[theorem]{Example}

\newenvironment{topbotframe} 
   {\needspace{2\baselineskip}\noindent\hrulefill\vspace{-0.1cm}}
   {\vskip -0.3cm\noindent\hrulefill\medskip\needspace{-2\baselineskip}}

\newcommand{\beq}{\begin{equation}}
\newcommand{\eeq}{\end{equation}}

\newcommand{\C}{\ensuremath{\mathbb{C}}}

\newcommand{\K}{\ensuremath{\mathbb{K}}}

\newcommand{\R}{\ensuremath{\mathbb{R}}}

\renewcommand{\Re}{\mathrm{Re}}
\renewcommand{\Im}{\mathrm{Im}}

\newcommand{\ket}[1]{|#1\rangle}
\newcommand{\bra}[1]{\langle#1|}

\newcommand{\Tr}{\mbox{\rm Tr}}
\newcommand{\Id}{\ensuremath{\mathrm{Id}}} % removed ``\mathop" since this is not an operator? otherwise it creates unwanted space
\newcommand{\linspan}{{\rm span}}

\newcommand{\HH}{\mathcal{H}}

\newcommand{\HA}{\mathcal{H}_A}
\newcommand{\HC}{\mathcal{H}_C}

\newcommand{\HB}{\mathcal{H}_B}

\newcommand{\eps}{\varepsilon}
\DeclareMathOperator{\poly}{poly}

\newcommand{\CHSH}{\textsc{CHSH}}
\newcommand{\HI}{\textsc{H}}

\newcommand{\setft}[1]{\mathrm{#1}}
\newcommand{\lin}[1]{\setft{L}\left(#1\right)}
\newcommand{\herm}[1]{\setft{H}\left(#1\right)}

\newcommand{\vecm}[2]{\overrightarrow{\setft{Mat}}_{#1}\left(#2\right)}

\newcommand{\matv}[2]{\overrightarrow{\setft{Mat}}_{#1}\left(#2\right)}

\newcommand{\obs}[1]{\setft{Obs}\left(#1\right)}

\newcommand{\nr}[1]{\left\|#1\right\|_\Psi}
\newcommand{\nrb}[1]{\big\|#1\big\|_\Psi}

\newcommand{\bias}{\omega}

\newcommand{\biasos}{\bias^{os}}
\newcommand{\biasnc}{\bias^{nc}}
\newcommand{\biasme}{\bias^{me}}
\newcommand{\biascom}{\bias^{sdp}}

\newcommand{\biasc}{\bias^{\C}}
\newcommand{\biasrk}{\bias^{rk1}}

\newcommand{\me}{\Psi^{me}}
\newcommand{\emb}{\Gamma}

\input{marginnotes}

\begin{document}

\title{\bf Quantum XOR Games}
\author{Oded Regev \footnote{CNRS, D{\'e}partement d'Informatique, {\'E}cole normale sup{\'e}rieure, Paris, and
Blavatnik School of Computer Science, Tel Aviv University.
Supported by a European Research Council (ERC) Starting Grant.
   } \and Thomas Vidick \footnote{Computer Science and Artificial Intelligence Laboratory, Massachusetts Institute of Technology. Supported by the National Science Foundation under Grant No. 0844626.}}
\date{}
\maketitle

%%%% DON'T REMOVE %%%%%
\noteswarning
%%%% DON'T REMOVE %%%%%

\vspace*{-0.5em}
\begin{abstract}
\input{sec_abstract.tex}
\end{abstract}

\onote{to check before submission: using ``bias" consistently and not ``value"; also ``rank-one" not ``rank one"}

\onote{remove vspace before abstract when not needed}

\tnote{Also: check naming of games, and that we never use ``partial transpose'' anymore}

\onote{check bibliography; BrietVidick will be published at some point; so is Buhrman et al.}

\onote{check for ``example" environments whose bottom line is orphaned}

\section{Introduction}\label{sec:intro}

\input{sec_intro.tex}

\section{Preliminaries}

\paragraph{Notation.} For an integer $n$, we use the notation $[n]$ to denote the set $\{1,\ldots,n\}$. For $x\in \R$ we let $\text{sign}(x) = x/|x|$ if $x\neq 0$, and $\text{sign}(0)=1$. If $x=(x_i)\in\R^n$ or $\C^n$, we let $\|x\|_\infty := \max_{i\in[n]} |x_i|$. For vectors $x,y\in\C^n$ we define their inner product $\langle x,y\rangle = \sum_i \overline{x_i}\ y_i$ and the norm $\|x\|=\langle x,x\rangle^{1/2}$. 

\paragraph{Matrices and norms.} A calligraphic letter $\HA$, $\HB$ will always denote a finite dimensional Hilbert space. $\lin{\HA,\HB}$ is the set of linear operators from $\HA$ to $\HB$, and $\lin{\HA}=\lin{\HA,\HA}$. $H(\HA)$ is the set of Hermitian operators on $\HA$, and $\obs{\HA}$ is the set of observables, i.e., Hermitian matrices whose eigenvalues are in $\{-1,1\}$. We use $M_n(\K)$ to denote $n\times n$ matrices over a field $\K$, and $M_n = M_n(\C)$. For $A\in \lin{\HA}$ we let $\|A\|_{\infty}$ be its operator norm (i.e., largest singular value) and $\|A\|_1 := \Tr\sqrt{A^\dagger A}$ its Schatten $1$-norm. 

\section{A review of classical XOR games}\label{sec:classical-xor}

\input{sec_classicalXOR.tex}

\section{Quantum XOR games}\label{sec:quantum-xor}

\input{sec_quantumXOR.tex}

\section{Some constructions}

\subsection{Rank-one quantum games}\label{sec:rank-one}

\input{sec_rankone.tex}

\subsection{The family \texorpdfstring{$(T_n)$}{Tn}}\label{sec:partial-transpose}

\input{sec_pt.tex}

\subsection{The family \texorpdfstring{$(\HI_n)$}{Hn}}\label{sec:hi-game}

\input{sec_hi.tex}

%\section{Open problems}\label{sec:open}

%\appendix

\section{Grothendieck inequalities}\label{sec:grothendieck}

\input{sec_grothendieck.tex}

%\section{Omitted proofs}\label{sec:omitted}
%\input{sec_omitted}

\bibliographystyle{alphaabbrvprelim}
\bibliography{oss}

\notesendofpaper

\end{document}

%% file: marginnotes.tex
\newif\ifnotes\notesfalse

\ifnotes
\usepackage{color}
\definecolor{mygrey}{gray}{0.50}
\newcommand{\notename}[2]{{\textcolor{mygrey}{\footnotesize{\bf (#1:} {#2}{\bf ) }}}}
\newcommand{\noteswarning}{{\begin{center} {\Large WARNING: NOTES ON}\endnote{Warning: notes on}\end{center}}}
\newcommand{\notesendofpaper}{{\theendnotes}}
\newcommand{\pnote}[1]{{\endnote{#1}}}

\else

\newcommand{\notename}[2]{{}}
\newcommand{\noteswarning}{{}}
\newcommand{\notesendofpaper}{}
\newcommand{\pnote}[1]{}

\fi

\newcommand{\tnote}[1]{{\notename{Thomas}{#1}}}
\newcommand{\onote}[1]{{\notename{Oded}{#1}}}

%% file: sec_abstract.tex
We introduce quantum XOR games, a model of two-player one-round games that extends the model of XOR games by allowing
the referee's questions to the players to be quantum states. 
We give examples showing that quantum XOR games exhibit a wide range of behaviors that are known not to exist for standard XOR games,
such as cases in which the use of entanglement leads to an arbitrarily large advantage over the use of no entanglement. By invoking two deep extensions of Grothendieck's inequality, we present an efficient algorithm that gives a constant-factor approximation to the best performance players can obtain in a given game, both in case they have no shared entanglement and in case they share unlimited entanglement. As a byproduct of the algorithm we prove some additional interesting properties of quantum XOR games, such as the fact that sharing a maximally entangled state of arbitrary dimension gives only a small advantage over having no entanglement at all.

%% file: sec_intro.tex
Two-player games play a central role in both computational complexity and quantum information theory. In the former, they crucially appear in major developments such as interactive proof systems~\cite{BenGolKilWig88STOC}, efficient proof verification~\cite{BabForLun91CC}, the PCP theorem~\cite{AroLunMotSudSze98JACM,AroSaf98JACM}, and hardness of approximation~\cite{FGLSS96}. In the latter, they are a powerful tool to quantify the power of entanglement~\cite{Bell:64a} and suggest experiments that demonstrate its nonlocal properties. 

In a two-player one-round game, a referee interacts with two players who cooperate in order to win the game. The referee chooses a pair of questions $(s,t)$ according to a publicly known distribution $\pi$ and sends one question to each player. The players are each requested to provide answers $a,b$ respectively. The players win or lose the game based on a public predicate $V(a,b|s,t)\in\{0,1\}$. Crucially, the players are not allowed to communicate between themselves.

The no-communication condition is traditionally interpreted as saying that the players can be modeled by a pair of functions $A:s \mapsto a$ and $B:t \mapsto b$.\footnote{One can extend this by allowing players to behave randomly, or even allow them access to a shared random string. However, a simple convexity argument shows this buys them no power, and optimal classical strategies are without loss of generality deterministic.} This implicit assumption, however, is challenged by quantum information theory. Indeed, quantum mechanics allows for an additional resource to be shared between the players: quantum entanglement. While shared entanglement does not allow for communication between the players, it has been known since the work of Bell~\cite{Bell:64a} that it can improve the players' success probability in such games. 
Examples of games in which entanglement allows for winning probabilities that are greater than what can be achieved by unentangled players are known in the quantum information literature as ``Bell inequality violations". Apart from their inherent theoretical interest, the existence of such games enables an experimental demonstration of the nonlocal nature of entanglement. 

\paragraph{XOR games.}
This paper is concerned with arguably the simplest type of two-player games, called XOR games. Those are two-player one-round games in which the referee's behavior is restricted: each player only provides him with a one-bit answer, and he is constrained to make his accept/reject decision based on the parity of the two bits alone. 

In the context of XOR games it is customary to quantify the players' success through their \emph{bias}, defined as twice  the difference between the players' success probability and their success probability if they answered all questions randomly. Optimizing the bias over players restricted to sharing a certain type of nonlocal resource leads to different quantities of interest. The \emph{unentangled bias} $\bias(G)$ corresponds to the largest possible bias achievable by players restricted to not using any entanglement at all. The \emph{entangled bias} $\bias^*(G)$ corresponds to players who may share an arbitrary entangled state. In addition, it will be interesting to consider the \emph{maximally entangled bias}, $\biasme(G)$, which corresponds to players restricted to sharing a maximally entangled state of arbitrary dimension. 

XOR games were first introduced in quantum information theory by Cleve et al.~\cite{CHTW04}, although they already appeared implicitly in the works of Bell~\cite{Bell:64a} and Clauser et al.~\cite{ClauserHSH:1969}. The first systematic study of such games from a mathematical point of view was undertaken by Tsirelson in the 80s. Tsirelson's key observation~\cite{Tsirelson:80a} is that the entangled bias can be exactly reformulated as a simple optimization problem over inner products of vectors in a space of bounded dimension (depending only on the size of the game). As observed in~\cite{CHTW04}, an immediate corollary of this characterization is that the entangled bias $\bias^*(G)$ can be computed efficiently using semidefinite programming techniques. An additional consequence, also due to Tsirelson, is that one can always achieve the optimum bias using a maximally entangled state, i.e., $\biasme(G) = \bias^*(G)$ for all XOR games $G$, and moreover, a maximally entangled state of relatively small dimension suffices. 

Using this reformulation, Tsirelson established a deep connection between XOR games and \emph{Grothendieck's inequality}~\cite{Gro53}, a fundamental inequality in Banach space theory. As a consequence, he showed that players using entanglement in an XOR game could only achieve a constant factor advantage over unentangled players --- the constant being \emph{Grothendieck's constant}. In other words, $\bias^*(G)$ is always at most a constant factor larger than $\bias(G)$. This established one of the first systematic limitations on the strength of quantum entanglement, and bounds on the entangled bias are now known as \emph{Tsirelson inequalities}. 

XOR games have also been studied extensively in theoretical computer science. In contrast to the entangled setting, H{\aa}stad~\cite{Hastad01} showed that it is NP-hard to approximate, within a small constant, the unentangled bias $\bias(G)$ of an XOR game $G$. From that result he deduced the NP-hardness of approximating the MAXCUT problem to within a constant factor (among others). 
The connection between XOR games and Grothendieck's inequality discovered by Tsirelson has also found applications in this context. For example, Alon and Naor~\cite{AN06} use the inequality, together with the observation mentioned above that one side of it can be efficiently computed, to obtain a constant-factor approximation algorithm for the problem of computing the cut-norm of a matrix. 
%Tsirelson's correspondence also results in an exact characterization of the maximum success probability of entangled strategies as the optimum of a semidefinite program. This characterization is in great part responsible for the success of XOR games, turning an a priori intractable problem (indeed, entangled strategies may be of arbitrarily large dimension, and it is not clear in general how or if one can perform optimization over them) into a tractable one. In fact, Tsirelson showed an even stronger result: given an optimal assignment of vectors to the left-hand side of~\eqref{eq:groth-ineq}, one may efficiently round the vectors to an entangled strategy for the players that achieves the same bias. This rounding has important consequences, showing for instance that optimal strategies can always be based on the use of a finite-dimensional maximally entangled state. 

\subsection{Our results}\label{sec:ourresults}

As described above, XOR games are quite well understood, an understanding to a large extent due to their elegant connection to semidefinite programming and Grothendieck's inequality. Unfortunately, their simple structure also means that the kind of behaviors they can exhibit are somewhat limited. For instance, as mentioned above, entanglement can only provide a relatively modest advantage over unentangled strategies. Moreover, one can always achieve the optimum winning probability using
a maximally entangled state of relatively small dimension.
  
A considerable amount of work in recent years has tried to identify games that exhibit a richer behavior (e.g.,~\cite{LTW08,KempeRT08,junge&palazuelos:largeviolation,PerezGarcia09arxiv,PalV10,BuhrmanRSW11,Regev11}). However, with the exception of the more recent~\cite{cooney}, which we discuss in more detail below (and perhaps~\cite{KempeRT08}), most of these papers focus on the analysis of \emph{specific} games, for which results such as large quantum-classical gaps are reported. Indeed, at this point we only have few tools (e.g.,~\cite{DLTW08}) that cut across large families of games and would enable one to prove general structural results.

\paragraph{Quantum XOR games.} In this paper we extend the framework of XOR games by allowing the referee's questions to the players to be quantum states. In a \emph{quantum} XOR game --- so named to differentiate them from the \emph{classical} XOR games discussed above --- the referee first chooses an index $i$ according to a public distribution $\pi$, based on which he prepares a bipartite state $\rho_i$ (whose description is also public). He sends one half of $\rho_i$ to the first player, and the other half to the other player. As in a classical XOR game, the players are required to reply with a single classical bit each, and the referee is restricted to base his accept/reject decision solely on the parity of the two bits he receives as answers. 

By considering the states $\rho_{s,t} = \ket{s}\bra{s}\otimes \ket{t}\bra{t}$, one immediately sees that quantum XOR games contain classical XOR games as a special case. Our results demonstrate that quantum XOR games are a fruitful generalization in two complementary ways. First, by providing examples of games exhibiting properties of entanglement that could not be observed in the context of classical XOR games, we show that quantum XOR games are a \emph{richer} model. Second, we show that, in spite of this greater generality, quantum XOR games remain a \emph{tractable} model. In particular, we give an efficient approximation algorithm for the players' maximum success probability, something that is known to exist only in a handful of other settings~\cite{CHTW04,KempeRT08} --- none of which is known to be as rich as that of quantum XOR games. These two aspects, we believe, make quantum XOR games a very attractive class of games to study.

\paragraph{A first example: the family $(T_n)$.} For any $n\geq 1$, let $T_n$ be the quantum XOR game in which the players are sent either of the two states
$$ \ket{\psi_0} = \frac{1}{\sqrt{2}}\ket{0}\ket{0}+ \frac{1}{\sqrt{2n}}\sum_{i=1}^n \ket{i}\ket{i}\qquad\text{and}\qquad \ket{\psi_1} = \frac{1}{\sqrt{2}}\ket{0}\ket{0}- \frac{1}{\sqrt{2n}}\sum_{i=1}^n \ket{i}\ket{i},$$
each chosen with probability $1/2$ by the referee, and are asked to produce answers with even parity in case the state is $\ket{\psi_0}$, and odd parity in case it is $\ket{\psi_1}$.
Even though the two states are orthogonal, it is not a priori clear how well the players can perform in this game: can $\ket{\psi_0}$ and $\ket{\psi_1}$ be \emph{locally} distinguished? Interestingly, the answer to this question crucially depends on the resources allowed for the players. The maximum bias achievable by players who do not share any entanglement is exactly $\bias(T_n) = 1/\sqrt{n}$. In fact, even players allowed to share an arbitrary supply of EPR pairs cannot do better: $\biasme(T_n)=1/\sqrt{n}$. Surprisingly, in case the players have access to an unrestricted amount of entanglement, we have $\bias^*(T_n)=1$: an unbounded advantage over the unentangled case. 
This is in stark contrast with the setting of classical XOR games, for which, as we already saw earlier, entangled players can only achieve a bounded advantage over unentangled players. Finally, we will also show that the optimal bias $\bias^*(T_n)=1$ can only be achieved in the limit of infinite entanglement, whereas for any classical XOR game there is an optimal strategy using an entangled state of bounded dimension depending only on the size of the game.

\paragraph{Algorithms.} The example of the games $(T_n)$ as well as further examples below demonstrate the richness of the model of quantum XOR games. Remarkably, despite encompassing such a wide variety of behaviors, quantum XOR games remain a tractable model, as is demonstrated by our main theorem.

\begin{theorem}\label{thm:algorithm}
There exists a polynomial-time algorithm which, given as input an explicit description of a quantum XOR game $G$, outputs two numbers $\biasnc(G)$ and $\biasos(G)$ such that
\begin{equation}\label{eq:main-bounds}
\bias(G) \leq \biasme(G) \leq \biasnc(G) \leq 2\sqrt{2} \,\bias(G) \qquad\text{and}\qquad  \bias^*(G) \leq \biasos(G) \leq 2 \,\bias^*(G).
\end{equation}
\end{theorem}

Thus, despite the fact that the entangled bias and the unentangled bias can differ greatly, both have nontrivial efficient approximations. We are not aware of any other model with this property. Moreover, both $\biasnc$ and $\biasos$ are expressible as the optimum of a polynomial-sized semidefinite program.
This property might aid in finding games that exhibit large separations between entangled and unentangled biases, say for the purposes of experimental demonstrations: given a candidate game $G$, run the algorithm above to approximate the gap between the two biases. 
We note that in addition to outputting the numerical values, the algorithm can also output descriptions of strategies satisfying the last inequality in each chain of inequalities. For details, see the formal statement of the main theorem in Theorems~\ref{thm:unentangled-bias} and~\ref{thm:entangled-bias}.

We emphasize that as is often the case, the existence of an efficient algorithm for a non-trivial problem allows one to derive surprising non-algorithmic conclusions. For instance, as an immediate corollary of the first sequence of inequalities stated in the theorem we obtain that, for any quantum XOR game $G$, $\biasme(G)\leq 2\sqrt{2} \bias(G)$: maximally entangled states only provide a bounded advantage over no entanglement at all. This in contrast with the general entangled case: the family $(T_n)$ shows that in general $\bias^*(G)$ can be arbitrarily larger than $\bias(G)$. Hence in this setting maximally entangled states can be arbitrarily far from an optimal resource. Such a behavior was known before for specific games \cite{junge&palazuelos:largeviolation,Regev11} but not for such a wide family of games. Another easy corollary (using the formal statement in Theorem~\ref{thm:unentangled-bias} and the remark after Definition~\ref{def:bias-me}) is that there is always an entangled strategy using just one EPR pair that achieves bias $\biasme(G)/2$, which is slightly better than the bias $\biasme(G)/(2\sqrt{2})$ we know can be achieved using no entanglement at all. Regarding the second sequence of inequalities in~\eqref{eq:main-bounds}, as a by-product of their proof we obtain that for any fixed $\eps$, there is an entangled strategy achieving bias that is at least $1/(2+\eps)$ of the optimum, and using only $O(\log n)$ qubits of entanglement, where $n$ is the dimension of each player's question. (See Theorem~\ref{thm:entangled-bias} for details.) 
We do not know how to prove any of the facts mentioned above without going through the a priori unrelated quantities $\biasnc$ and $\biasos$.

\paragraph{Techniques: Grothendieck inequalities.}  
Our main result, Theorem~\ref{thm:algorithm}, is proved by establishing a strong connection between quantum XOR games and two deep extensions of Grothendieck's inequality.
The first extension, which is used to prove the first sequence of inequalities surrounding $\biasnc$, is known as the \emph{non-commutative Grothendieck inequality}. The inequality, already conjectured by Grothendieck~\cite{Gro53}, was proved by Pisier~\cite{Pisier78NCGT} and then in a more general form by Haagerup~\cite{Haagerup85NCGT}. The second one, which is used to prove the second sequence of inequalities surrounding $\biasos$, is known as the \emph{operator space Grothendieck inequality} and was proved by Pisier and Shlyakhtenko~\cite{PS02OSGT} and by Haagerup and Musat~\cite{HM08}.\footnote{See also~\cite{RegevV12b} for a recent alternative proof inspired by quantum information theory, and more specifically the role played by the embezzlement state in the analysis of the entangled bias of the family $(T_n)$.}

Most of the effort in establishing our main theorem goes into interpreting these Grothendieck inequalities as statements relating biases (unentangled bias in the former case, and entangled bias in the latter) to semidefinite programs ($\biasnc$ in the former case, and $\biasos$ in the latter case), for which efficient algorithms are known. Our results give the first application of these inequalities to quantum information theory.\footnote{See however ``related work'' below for a discussion of concurrent work by Cooney et al.~\cite{cooney}, who independently found another application of the operator space Grothendieck inequality.} 
Much of the mathematical literature in this area can be intimidating at first (a good starting point is a recent survey by Pisier~\cite{PisierGT}), and we hope that our self-contained presentation will contribute to promoting those inequalities as powerful tools in complexity theory and quantum information theory, and will lead to further applications.

\pnote{To any quantum XOR game $G$ we may associate a bilinear form $\varphi_G : M_n(\C)\times M_n(\C)\to \C$ which takes as input a pair $(A,B)$ of observables (or, more generally, complex matrices) on the players' message spaces, and outputs the bias that (unentangled) players applying $A,B$ would obtain in $G$. With this definition the norm of $\phi_G$, by definition the supremum of $\varphi_G(A,B)$ over all choices of $A,B$ with operator norm at most $1$, coincides with the complex bias of $G$: $\|\varphi_G\| = \bias^\C(G)$. This simple equation is interesting because it relates a fundamental operational quantity associated to the game $G$ with a natural mathematical quantity. Other norms associated to a bilinear form $\varphi$ are of mathematical interest. In particular, its \emph{completely bounded} norm $\|\varphi\|_{cb}$ is defined as the supremum over all integers $d$ of the norm of the \emph{amplification} $\varphi\otimes \Id_d : M_{nd}(\C) \times M_{nd}(\C) \to \C$. One can verify that this norm exactly correspond to the entangled bias $\bias^*(G)$: the connection between the game $G$ and the bilinear form $\varphi_G$ extends to the setting of entangled players. In fact, a third norm on $\varphi$, its \emph{tracially bounded norm} $\|\varphi_G\|_{tb}$, also arises naturally in some contexts, and turns out to equal the maximally entangled bias $\biasnc(G)$.

Abstractly, different variants of Grothendieck's inequality relate either of the norms introduced above to other quantities associated to a given bilinear form $\varphi:E\times F \to \C$. Defining these quantities here would take us too far, and we refer to the paper for more details. The key point for us is that, while in general $\|\varphi_G\|$ and $\|\varphi_G\|_{cb}$ may be quantities that are hard to compute, it turns out that the quantity to which Grothendieck's inequality relates each of them is efficiently computable. It is this observation which is key to Theorem~\ref{thm:algorithm}: in our language, the two variants of Grothendieck's inequality we consider~\cite{Haagerup85NCGT,HM08} state that either of the norms above is within a factor $2$ of an efficiently computable quantity.}

\paragraph{Families of quantum XOR games.}
Our second main contribution consists in introducing and analyzing in detail two specific families of quantum XOR games. These families simultaneously demonstrate the flexibility of the model of quantum XOR games and illustrate the inequalities in Theorem~\ref{thm:algorithm}.

The first family is the family of games $(T_n)$ already mentioned above.\footnote{This family of games was suggested to us by David P\'erez-Garc\'ia~\cite{david:transpose}, and can be seen to correspond to a certain natural map in operator space theory, namely the identity map $id:\,R_n\to C_n$, where $R_n$ is the ``row'' operator space and $C_n$ the ``column'' operator space. Some of its properties mentioned below are related to the fact that it has norm $1$ but ``completely bounded'' norm $\sqrt{n}$.} This family demonstrates the possibility of obtaining an unbounded gap between the entangled and unentangled settings, implying in particular that the two sequences of inequalities in~\eqref{eq:main-bounds} cannot be merged into a single one. 
The next theorem summarizes the properties of this family. (The definition of $\biasc$ will be given later.)

\begin{theorem}\label{thm:ptgamemain}
Let $n\geq 1$. Then
$$\bias(T_n) \,=\, \biasc(T_n) \,=\, \biasme(T_n)\,=\, \biasnc(T_n) \,=\,\frac{1}{\sqrt{n}},$$
and
$$ \bias^*(T_n) \,=\, \biasos(T_n) \,=\, 1.$$
Moreover, the perfect winning probability $\bias^*(T_n)=1$ is only achieved in the limit of infinite entanglement: for any fixed $n \ge 2$, players sharing an arbitrary finite-dimensional state cannot win the game $T_n$ with certainty. 
\end{theorem}

The two sequences of equalities are proven by direct calculation. As a result, we determine optimal strategies both in the unentangled and entangled cases. These strategies are relatively simple and low-dimensional, and for moderate values of $n$ the game $T_n$  may provide a good candidate for the experimental demonstration of the nonlocality of entanglement.
The moreover part of the theorem is proved by observing that the games $T_n$ are closely related to the ``coherent state exchange'' game introduced in~\cite{LTW08}.

In analogy with the classical setting, and in light of the above theorem, one may expect that the inequalities $\biasme(G)\leq \biasnc(G)$ and $\bias^*(G)\leq \biasos(G)$ in~\eqref{eq:main-bounds} should, in fact, be equalities.\footnote{Since quantum XOR games generalize classical XOR games, we know that the inequality $\bias(G)\leq \biasnc(G)$ is \emph{not} always an equality.} Nevertheless, our second family of games, the games $(\HI_n)$, whose properties are summarized in the following theorem, shows that equalities do \emph{not} hold in general: there are games for which $\biasos$ is strictly greater than $\bias^*$, and also $\biasnc$ is strictly greater than $\biasme$. 

\begin{theorem}\label{thm:higamesummary}
There exists a family of games $(\HI_n)$ for which the following hold: 
$$\frac{2}{5}\,=\,\bias(\HI_1)\,=\,\bias^\C(\HI_1) \,<\, \biasme(\HI_1)\,\le\, \bias^*(\HI_1)  \,<\, \biasnc(\HI_1)\,=\,\biasos(\HI_1)\,=\,\frac{3}{5}$$
and for all $n\geq 1$
$$\bias(\HI_n)\,= \,\bias^\C(\HI_n)\,=\,\frac{n+1}{2n+1}\biasnc(\HI_n) .$$
\end{theorem}

This family of games is related to the CAR algebra. Up to an unimportant scaling, the game $\HI_1$ which plays a particularly important role above can be described concretely as follows. 
The referee first picks two distinct integers $j<k\in \{1,2,3\}$ uniformly at random. He then sends one of the two states $\frac{1}{\sqrt{2}}\big( \ket{j}\pm i\ket{k}\big)$, again uniformly at random, to each player (so each of the four possible combinations arises with probability $1/4$). The referee accepts the players' answers $a,b\in\{0,1\}$ if and only if $a\oplus b = 1$ in case they were both sent ``$+$'', or both ``$-$'', states, and $a\oplus b=0$ otherwise. 
We note that this family already appears in the literature, albeit in the language of operator spaces.  
It first appeared in~\cite{Blecher88}, and was later investigated in depth by Haagerup and Itoh~\cite{HI95}, who already
proved many of the statements in the above theorem. Our main contribution here is the bound $\bias^*(\HI_1)< 3/5$, which improves on the weaker bound $\biasme(\HI_1)  < 3/5$ already appearing in~\cite{HI95}. Moreover, their proof is based on the use of ultrafilters, and as such is non-explicit and relies on the axiom of choice; in contrast, our proof is more direct and quantitative. 

Finally, we briefly note that another interesting family of games, the games $(C_n)$, was introduced in~\cite[Section~5]{cooney} to show that the entangled bias of rank-one quantum games (see Section~\ref{sec:related} below and Section~\ref{sec:rank-one} for more details) does not obey a strong parallel repetition theorem. Translated to a quantum XOR game, $C_n$ essentially corresponds to the following game. The referee chooses a random integer $k\in\{1,\ldots,n\}$, and sends one of the two states $(\ket{0,k}\pm\ket{k,0})/\sqrt{2}$, each chosen with probability $1/2$, to the players. They should produce answers with even parity in case they were sent a ``$+$'' state, and odd in case it was a ``$-$'' state. Although we will not prove them here, the following equalities either follow from the results of~\cite[Section~5]{cooney}, or can be given a direct proof:
$$ \bias(C_n) \,=\,  \biasos(C_n) \,=\, \frac{1}{n} \qquad\text{and}\qquad \bias(C_n\otimes C_n) \,\geq\, \frac{1}{2n}.$$
The values of $\biasme(C_n),\biasnc(C_n)$ and $\bias^*(C_n)$ can be deduced automatically from the equalities above using the inequalities in Theorem~\ref{thm:algorithm}. 
This family of examples (for $n>2$) shows that none of the quantities we introduce for quantum XOR games satisfies a perfect parallel repetition theorem. This is perhaps somewhat surprising since classical XOR games do satisfy perfect parallel repetition~\cite{CSUU08}.
\pnote{An example from~\cite[Theorem 1.2]{cooney} shows that parallel repetition does not hold for either the nc or os biases. For the nc bias this is a bit surprising, since the SDP has the form required to apply the Mittal-Szegedy results. However, the product SDP will not correspond to the tensor-product XOR game. It will have additional constraints mounding the ``mixed row-column'' norm of $X,Y$ (it will have four constraints on each of $X,Y$ instead of just two), and these additional constraints are not satisfied by the strategy in the example below. 

Here is a rough sketch. The game has matrix $M = 1/(2n) \sum_k E_{1k} \otimes E_{k1} + E_{k1} \otimes E_{1k}$, which is Hermitian. First we upper bound $\biasos(M)$. Let $(X_i),(Y_i),(t_i)$ be an SDP solution. Write
\begin{align*}
\sum_i \Tr( M X_i\otimes Y_i) &= \frac{1}{2n}\sum_{i,k} \big(\bra{k}X_i\ket{1}\bra{1}Y_i\ket{k} + \bra{1}t_iX_i\ket{k}\bra{k}t_i^{-1}Y_i\ket{1}\\
&\leq \frac{1}{2n}\Big( \sum_{i,k} |\bra{k}X_i\ket{1}|^2\Big)^{1/2}\Big( \sum_{i,k} |\bra{1}Y_i\ket{k}|^2\Big)^{1/2}+\frac{1}{2n}\Big( \sum_{i,k} |\bra{1}t_iX_i\ket{k}|^2\Big)^{1/2}\Big( \sum_{i,k} \bra{k}t_i^{-1}Y_i\ket{1}\Big)^{1/2} \\
&= \frac{1}{2n}\Big(  \bra{1} \sum_{i} X_i^\dagger X_i\ket{1}\Big)^{1/2}\Big( \sum_{i} \bra{1}Y_iY_i^\dagger \ket{1}\Big)^{1/2}+\frac{1}{2n}\Big( \sum_{i} \bra{1} t_i^2 X_iX_i^\dagger\ket{1}\Big)^{1/2}\Big( \sum_{i} \bra{1}t_i^{-2}Y_i^\dagger Y_i\ket{1}\Big)^{1/2} \\
&\leq  \frac{1}{2n}\Big(\big\| \vec{X}^\dagger \vec{X} \big\|\big\| \vec{Y} \vec{Y}^\dagger \big\| + \big\|  \vec{tX}\vec{tX}^\dagger \big\|\big\|\vec{t^{-1}Y}^\dagger  \vec{t^{-1}Y} \big\|\Big)\\
&\leq \frac{1}{n}.
\end{align*}

Next we lower bound the unentangled bias of the repeated game, with matrix $M\otimes M$. Let 
$$X = Y= \Big(\sum_{k} E_{k1} \otimes E_{1k} + E_{1k} \otimes E_{k1}\Big)$$
$X$ and $Y$ are Hermitian and have operator norm at most $1$\onote{actually, for this we need to consider the modified game described below}. Their value in the game is
\begin{align*}
\Tr( (M\otimes M) (X\otimes Y)) &= \frac{1}{(2n)^2} \sum_k \Tr\big( (\ket{11ii}\bra{ii11} + \ket{ii11}\bra{11ii} + \ket{1ii1}\bra{i11i} + \ket{i11i}\bra{1ii1} )\\
&\qquad\qquad\qquad \cdot ( \ket{i1i1}\bra{1i1i} + \ket{1i1i}\bra{i1i1i} + \ket{i11i}\bra{1ii1} + \ket{1ii1}\bra{i11i})\big)\\
&= \frac{1}{2n},
\end{align*} 
where for the first equality we re-ordered the terms corresponding to $M$ so that the order of the tensors matched that of $X\otimes Y$. Hence $\bias(M\otimes M \geq 1/(4n)$.
Operationally, this can be described as follows. Each player, given his two inputs, performs an incomplete measurement containing the two-dimensional subspaces $span(\ket{1k},\ket{k1})$ for all $k \neq 1$. If the measurement does not fall in one of these subspaces, he outputs a random answer. Otherwise (and this happens with probability $1/2n$), we are left with the state $\ket{1kk1}\pm \ket{k11k}$ which after a local unitary transformation the users can convert to $T_1$ and win perfectly. 

Running this example through the SDP solver, we find the following. First, the solver returns $\biasos(M) = 1/n$ for $n$ up till $n=5$. It also returns $\biasos(M\otimes M) = 1/(2n)$ for $n=2,3$ (couldn't go further) and $\biasnc(M)=1/(2n)$ for $n$ up to to $n=5$.

We can also consider the same game, but without $k=1$, i.e. the game $C_n$ with matrix $M(C_n) = 1/(2n) \sum_{k=1}^n E_{0k} \otimes E_{k0} + E_{k0} \otimes E_{0k}$ described in the intro. For that game, we find $\biasnc(C_n) = \biasos(C_n) = 1/n$ for $n$ up till $n=7$, and $\biasnc(C_n\otimes C_n) = 1/(2n)(1+1/n)$ for $n=2,3,4$ ($\biasos$ gives the same value for $n=2$, but it won't run for larger $n$)
}

\subsection{Related work}\label{sec:related}

The model of two-player games in general, and its quantum information aspects in particular, have been widely studied in the past and we will not attempt to give a comprehensive survey, instead only focusing on the results most closely related to ours. Recently Buscemi~\cite{Buscemi11} considered a model he calls ``semi-quantum games", in which the players are sent arbitrary quantum states as questions, and their answers are arbitrary classical strings (hence semi-quantum games contain quantum XOR games as a subclass). He establishes an interesting connection between such games and the task of transforming one state into another using local operations and shared randomness (LOSR): such a task is possible if and only if players sharing the former state can always obtain an expected payoff that is at least as high as players sharing the latter, in any semi-quantum game. 

Quantum XOR games can also be interpreted as a particular formalization of a \emph{local} distinguishing task: indeed, any quantum XOR game can be thought of as a game in which the players are given one of two density matrices, and are asked to produce bits with even or odd parity depending on which state they were given. With the notable exception of~\cite{Buscemi11}, much of the literature in this area is concerned with LOCC (local operations and classical communication) distinguishability (see, e.g.,~\cite{BDFM99,Wat05,ChildsLMO12}), and thus does not seem directly related to our results. 

Recently, Cooney et al.~\cite{cooney} introduced another model of two-player games which they call \emph{rank-one quantum games}, in which, informally speaking, the players are sent parts of a pure state prepared by the referee, and are supposed to convert it to another pure state.
One of their main interests is in approximating the maximum success probability of arbitrary entangled players in their model, which they do using the operator space Grothendieck inequality, just as we do for our model. In fact, this is not a coincidence, since as we describe in Section~\ref{sec:rank-one}, there is an explicit connection between the two models. See also that section for more details on their model. The rest of their paper focuses on other questions not considered by us, such as that of parallel repetition. Most of our work was mainly done independently and concurrently to theirs, although we did benefit from communicating with them about their work, and we thank them for sharing it with us at early stages.

\subsection{Directions for future work}

\input{sec_open.tex}

\paragraph{Acknowledgments.} We are grateful to David P\'erez-Garc\'ia for suggesting the family of games $(T_n)$. We also thank him and Carlos Palazuelos for many useful discussions.

%% file: sec_open.tex
Our work leaves many questions open; we list just a few that we think are interesting and would deserve further exploration. 

\paragraph{Gaps between the biases.} Among the bounds that we proved between the different biases associated to a quantum XOR game, there are two that we do not know to be tight. First, we showed that $\biasnc$ is at most a $2\sqrt{2}$ factor larger than the unentangled bias $\bias$, but we only know of a factor $2$ separation, which follows from the family of games $(\HI_n)$. Second, we showed that $\biasos$ is at most a factor $2$ greater than the entangled bias $\bias^*$, but the best separation we can prove between the two is the one in Theorem~\ref{thm:higamesummary}, which is of a constant factor very close to $1$. Can that separation be improved? 

A related question is to study the gaps between the quantities $\biasnc$ and $\biasos$ that we introduce and the corresponding biases $\bias$ or $\biasme$ and $\bias^*$ in the regime where their value is close to $1$. In the case of classical XOR games it is known~\cite{CHTW04} that if the entangled bias is at least $1-\eps$ then the unentangled bias is at least $1-O(\sqrt{\eps})$. Could a similar result be shown between $\biasos$ and $\bias^*$, or between $\biasnc$ and $\bias$ or $\biasme$? (The example of the game $T_n$ shows that this \emph{does not} hold of the entangled and unentangled, or even maximally entangled, biases.)

Finally, it would be interesting to determine the maximum ratio achievable between, say, the entangled and unentangled biases of a given game, as a function of the size of the game or of the dimension of the entangled state used by the players in the entangled strategy. Such bounds are already known for three-player classical XOR games~\cite{perezgarcia:2008,BV12} and two-player games with arbitrary answer size~\cite{junge&palazuelos:largeviolation}.

\paragraph{Hardness results.} Results of H{\aa}stad~\cite{Hastad01} on classical XOR games imply that their unentangled bias, and by extension the unentangled bias of quantum XOR games, is NP-hard to approximate within small constant factors. What about the entangled bias? For classical XOR games it follows from Tsirelson's results that it can be computed efficiently. For quantum XOR games, the quantity $\biasos$ gives a factor $2$ approximation. Is there a better efficiently computable approximation, or can one perhaps show that the entangled bias is hard to approximate (possibly assuming the Unique Games conjecture~\cite{Khot02})?

\paragraph{Combinatorial applications.} The commutative Grothendieck inequality has been successfully used to devise constant-factor approximation algorithms for combinatorial problems such as computing the cut-norm of a matrix~\cite{AN06}. Could the non-commutative generalizations lead to new approximation algorithms for combinatorial problems, possibly by interpreting them as quantum XOR games? See~\cite{RegevV12c} for some recent work in this direction.

%% file: sec_classicalXOR.tex
In this section we review some definitions and results on two-player classical XOR games. Although most of them already appear in the paper by Cleve et al.~\cite{CHTW04}, our presentation is slightly different and is meant to ease the comparison with the case of quantum XOR games.

A \emph{classical XOR game} $G$ of size $n$ is specified by $n^2$ real coefficients $R = (R_{s,t})_{s,t\in[n]}$ satisfying the normalization condition $\sum_{s,t=1}^n |R_{s,t}| = 1$. The game is played as follows. The referee picks a pair of integers $(s,t)\in [n]^2$ according to the distribution $\{\pi(s,t) = |R_{s,t}|\}$, and sends $s$ to the first player, Alice and $t$ to the second player, Bob. Upon receiving their respective questions, the players each answer with a single bit $a,b\in\{0,1\}$. The referee accepts the players' answers if and only if $(-1)^{a\oplus b} = \text{sign}(R_{s,t})$. Notice that players sending random answers will be accepted with probability $1/2$ in $G$. The \emph{bias} $\bias(G)$ of $G$, defined as twice the difference between the maximum success probability of any players and the success probability of the random strategy (which is $1/2$ in this case), can then be formally expressed as
\begin{equation}\label{eq:class-bias}
\bias(G)\,=\,\bias(R(G))\,:=\, \max_{x_s,y_t\in\{\pm1\}}\,  \Big|\sum_{s,t} \,R_{s,t}\,x_s y_t\Big| .
\end{equation}
Note that the maximum on the right-hand side may be equivalently taken over all $x,y\in \R^n$ such that $\|x\|_\infty,\|y\|_\infty \leq 1$ (instead of over all $x,y \in \{-1,1\}^n$): the maximum will always be attained at an extreme point. 

\begin{topbotframe}
\begin{example}[The CHSH game]
The \emph{CHSH game} is a simple XOR game derived from the famous Bell inequality originally introduced by Clause, Horne, Shimony and Holt~\cite{ClauserHSH:1969}. It is a game of size $2$ with coefficients 
$$R_{11}=\frac{1}{4}, \quad R_{12}=\frac{1}{4}, \quad R_{21}=\frac{1}{4}, \quad R_{22}=-\frac{1}{4}.$$
It is not hard to verify that for this game the bias is $\bias({\CHSH}) = 1/2$. 
\end{example}
\end{topbotframe}

For any (possibly complex) $R$, we also consider the \emph{complex bias}, a quantity we will denote $\bias^\C(R)$ and define as  
\begin{equation} \label{eq:class-complex-bias}
\bias^\C(R)\,:=\, \max_{x_s,y_t\in \C, \, |x_s|,|y_t|\leq 1}\,  \Big|\sum_{s,t} \,R_{s,t}\,x_s y_t\Big|.
\end{equation}
Informally speaking, this can be thought of as allowing the players to respond not just with bits in $\{-1,1\}$ but rather with any complex number on the unit circle.
The complex bias can sometimes be larger than the bias, even for real coefficients $R$. The following example shows that it can be a factor $\sqrt{2}$ larger, and in Claim~\ref{claim:obs_value} in the next section we will show that a result of Krivine~\cite{Krivine:79a} implies that the inequality $\bias^\C(R)\leq \sqrt{2}\bias(R)$ holds for any real $R$.

\medskip 

%\noindent\begin{minipage}{\textwidth}
\begin{topbotframe}
\begin{example}[CHSH, complex bias]\label{ex:chsh-complex} The $\CHSH$ game satisfies $\bias^\C(\CHSH) = \sqrt{2}/2$. To show that the complex bias is at least $\sqrt{2}/2$, it suffices to use the modulus-$1$ complex numbers $x_1 = (1+i)/\sqrt{2}$, $x_2 = (1-i)/\sqrt{2}$, $y_1 = 1$ and $y_2 = -i$ in the right-hand side of~\eqref{eq:class-complex-bias}. The fact that $\bias^\C(\CHSH)\leq \sqrt{2}/2$ will follow from the bound on $\biascom(\CHSH)$ derived in Example~\ref{ex:chsh-sdp} below.
\end{example}
\end{topbotframe}
%\end{minipage}

The maximization on the right-hand side of~\eqref{eq:class-bias} is a quadratic optimization problem. Given an XOR game $G$ the problem of computing, or even approximating within a small constant factor, the quantity $\bias(G)$ was shown NP-hard by H{\aa}stad~\cite{Hastad01}. However, $\bias(G)$ may be bounded from above by the following natural relaxation of~\eqref{eq:class-bias}: 
\begin{equation}\label{eq:class-bias-sdp}
 \bias(G)\,\leq\,\bias^\C(R(G))\,\leq\,\biascom(R(G)) \,:=\, \sup_{d,\,x_s,y_t \in \C^d}\,\Big|\sum_{s,t} R_{s,t}\, \langle \overline{x}_s , y_t \rangle \Big|,
\end{equation}
where now the supremum is taken over all dimensions $d$ and vectors $x_s,y_t \in\C^d$ with Euclidean norm at most $1$.\footnote{It is not hard to see that in the case of real coefficients $R$, the supremum in~\eqref{eq:class-bias-sdp} can equivalently be taken over real vectors $x_s,y_t\in\R^d$.} The complex conjugation of $x_s$ in~\eqref{eq:class-bias-sdp} is somewhat unusual, but we introduce it for convenience and consistency with the rest of the paper; it clearly does not affect the optimization problem. 

The fact that this is a relaxation (i.e., the second inequality above) follows since a number of modulus at most $1$ is also a one-dimensional vector of norm at most $1$.
Moreover, it is easy to verify that the supremum above is a semidefinite program, and as such can be computed up to precision $\eps$ in time $\poly(n,\log 1/\eps)$.  In more detail, multiplying all $x_s$ by a complex phase if necessary, the absolute values on the right-hand side of~\eqref{eq:class-bias-sdp} can be replaced by the real part without changing the supremum. The resulting expression can be written as the maximization of a linear function of the inner products $\langle \overline{x}_s , y_t \rangle$, under constraints bearing on the inner products $\langle x_s , x_s \rangle$ and $\langle y_t , y_t \rangle$. Such an optimization problem can then be formulated as a real semidefinite program using standard techniques (see Section~4.6.2  in~\cite{BoydV04} for generalities on semidefinite programs, and Exercise~4.42 in particular for dealing with complex vectors). 

\begin{topbotframe}
\begin{example}[CHSH, bias of semidefinite relaxation]\label{ex:chsh-sdp} For the $\CHSH$ game we have $\biascom(\CHSH) = \sqrt{2}/2$. Indeed, a lower bound of $\sqrt{2}/2$ follows from the lower bound on the complex bias proved in Example~\ref{ex:chsh-complex}. A matching upper bound can be shown as follows: for any choice of unit vectors $x_1,x_2,y_1$ and $y_2$,
\begin{align*}
\frac{1}{4}\big|\langle \overline{x}_1, y_1\rangle + \langle \overline{x}_2, y_1\rangle + \langle \overline{x}_1, y_2\rangle - \langle \overline{x}_2, y_2\rangle \big| &=\frac{1}{4}\big| \langle \overline{x_1+ x_2}, y_1\rangle + \langle \overline{x_1 - x_2}, y_2 \rangle \big|\\
&\leq \frac{1}{4}\big(\|x_1+ x_2\| + \|x_1 - x_2\|\big) \\
&\leq\frac{\sqrt{2}}{4}\big( \|x_1+ x_2\|^2 + \|x_1 - x_2\|^2\big)^{1/2}\,\leq\,\frac{\sqrt{2}}{2}.
\end{align*}
\end{example}
\end{topbotframe}

How good is the approximation of $\bias(G)$ by $\biascom(G)$? Example~\ref{ex:chsh-sdp} above shows that $\biascom(G)$ can be at least a factor $\sqrt{2}$ larger than $\bias(G)$. As it turns out, this is not far from the worst that can happen: the relaxation~\eqref{eq:class-bias-sdp} is always at most a small constant factor larger than $\bias(G)$. This is essentially the essence of Grothendieck's inequality~\cite{Gro53}. We will discuss that inequality further in Section~\ref{sec:grothendieck}; in the present context, it directly implies the following.

\begin{theorem}\label{thm:class-bias} Let $n$ be any integer and $R=(R_{s,t})_{s,t\in[n]}$ real coefficients. Then
$$ \bias(R) \,\leq \,\biascom(R) \,\leq \, K_G^\R\,\bias(R),$$
where $K_G^\R$ is the so-called real Grothendieck constant which is known to satisfy $K_G^\R\leq 1.782\ldots$~\cite{Kri73,BMMN11}. Moreover, for any $R$ with complex coefficients,
$$ \bias^\C(R) \,\leq \,\biascom(R) \,\leq \, K_G^\C\,\bias^\C(R),$$
where $K_G^\C \leq 1.405\ldots$~\cite{Haa87} is the complex Grothendieck constant.
\end{theorem}

Next, we consider the case in which players are allowed to share an arbitrary state $\ket{\Psi}$, leading to the definition of the \emph{entangled bias}, 
\begin{equation}\label{eq:class-bias-ent}
 \bias^*(G) \,=\, \bias^*(R(G))\,:=\, \sup_{d,A_s,B_t,\ket{\Psi}}\,\Big|\sum_{s,t} \,R_{s,t}\, \bra{\Psi} A_s\otimes B_t \ket{\Psi}\Big|,
\end{equation}
where here the supremum is taken over all dimensions $d$, sequences of matrices $A_s,B_t\in\herm{\C^d}$ of operator norm at most $1$, and states $\ket{\Psi} \in \C^d\otimes \C^d$. By linearity, the supremum could equivalently be taken over all $A_s,B_t \in \obs{\C^d}$ without changing its value.

\begin{topbotframe}
\begin{example}[CHSH, entangled bias]\label{ex:chsh-entangledbias} 
The entangled bias for the CHSH game satisfies $\bias^*({\CHSH}) = \sqrt{2}/2$. Indeed, one can first verify that  the following strategy for the players achieves a bias of $\sqrt{2}/2$: Alice and Bob share a single EPR pair $\ket{\Psi} = (\ket{00}+\ket{11})/\sqrt{2}$. Upon receiving her question $s$, Alice measures either in the computational $(s=1)$ or the Hadamard $(s=2)$ basis. Bob measures in the computational basis rotated by either $\pi/8$ $(t=1)$ or $3\pi/8$ $(t=2)$. Moreover, the bound $\biascom(\CHSH)\leq \sqrt{2}/2$ given in Example~\ref{ex:chsh-sdp}, together with Lemma~\ref{lem:bias-relaxation-classical} below, which shows that $\biascom$ is always an upper bound on $\bias^*$, imply that this is best possible.
\end{example}
\end{topbotframe}

While a priori bounds on the entangled bias (i.e., Tsirelson inequalities) may not be easy to obtain (indeed, the supremum on the right-hand side of~\eqref{eq:class-bias-ent} extends to spaces of arbitrary dimension), Tsirelson showed that, somewhat surprisingly, the relaxation~\eqref{eq:class-bias-sdp} is \emph{also} a relaxation of the entangled bias. We include the short proof, as we will later extend it to the setting of quantum XOR games.

\begin{lemma}[Tsirelson~\cite{Tsirelson:85b}]\label{lem:bias-relaxation-classical} For any real $R$, 
$$\bias^*(R)\,\leq\,\biascom(R).$$
\end{lemma}

\begin{proof}
Let $(X_s,Y_t,\ket{\Psi})$, where $X_s,Y_t\in\obs{\C^d}$ and $\ket{\Psi}\in \C^d\otimes \C^d$ is a unit vector, be an arbitrary strategy for the players. Up to a local rotation of Alice's and Bob's private spaces we may write the Schmidt decomposition $\ket{\Psi} = \sum_{i=1}^d \lambda_i \ket{i} \ket{i}$, so that the bias achieved by this strategy is
\begin{align*}%\label{eq:class-sdp-0}
 \sum_{s,t}  R_{s,t}\, \sum_{i,j}\lambda_i \lambda_j \bra{i}X_s\ket{j} \bra{i}Y_t \ket{j} .
\end{align*}
For any $s,t$ we have that $\sum_{i,j}\lambda_i \lambda_j \bra{i}X_s\ket{j} \bra{i}Y_t \ket{j} = \langle \overline{x}_s, y_t \rangle $ where $x_s$ and $y_t$ are the $d^2$-dimensional vectors given by
$$ x_s \,:=\, \big( \lambda_i \bra{i}X_s\ket{j} \big)_{i,j}\qquad\text{and} \qquad y_t\,:=\,\big(\lambda_j \bra{i}Y_t\ket{j}\big)_{i,j}.$$
The vector $x_s$ consists of the $d^2$ entries of $X_s$ after weighing row $i$ by $\lambda_i$; similarly $y_t$ consists of the entries of $Y_t$, where this time we weigh the column $j$ by $\lambda_j$.
Note that the particular weighing scheme we chose is arbitrary, and we could also have decided to weigh the columns of $X_s$ and the rows of $Y_t$. The important point is that, since an observable has all of its rows and columns of norm $1$, both $x_s$ and $y_t$ have norm $1$. 
 Hence the collection $\{{x}_s,{y}_t\}$ constitutes a feasible solution to the right-hand side of~\eqref{eq:class-bias-sdp}, proving $\biascom(R)\geq\bias^*(R)$. 
\end{proof}

In fact, Tsirelson showed more: for any XOR game $G$, with coefficients $R$, the quantities $\bias^*(R)$ and $\biascom(R)$ are equal! That is, the relaxation of the unentangled bias that we introduced in~\eqref{eq:class-bias-sdp} exactly corresponds to the maximum bias achievable using arbitrary entangled strategies. 
Moreover, Tsirelson showed that the optimum bias is always achievable using a particular state, the maximally entangled state 
$$\ket{\Psi_d^{me}} \,:=\, \frac{1}{\sqrt{d}} \sum_{i=1}^d\, \ket{i}\ket{i}.$$
Denoting $\biasme(R)$ the largest bias achievable by players who are restricted to using entanglement of the form $\ket{\Psi^{me}}$ (we will leave the dimension subscript $d$ implicit whenever it is unrestricted), we have the following. 

\begin{proposition}[Tsirelson]\label{prop:tsirelson} For any real $R$, the following inequalities hold
$$ \bias(R)\,\leq\, \biasme(R)\,=\,\bias^*(R)\,=\,\biascom(R)\,\leq K_G^\R \,\bias(R). $$
\end{proposition}

In addition, Tsirelson~\cite[Lemma~3.1]{Tsirelson:85b} showed that the optimal bias $\bias^*(R)$ could be achieved using a maximally entangled state of dimension at most $2^{O(\sqrt{n})}$, where $n$ is the size of the game. We refer the reader to~\cite{Slofstra11} for additional results on the amount of entanglement required to play XOR games (near-)optimally.

%% file: sec_quantumXOR.tex
In this section we formally introduce quantum XOR games and prove our main theorem, Theorem~\ref{thm:algorithm}, together with the extensions that were discussed in the introduction. We start by defining quantum XOR games in Section~\ref{sec:quantum-defs}, and state several equivalent operational interpretations of the definition. 
In Section~\ref{sec:unentangled-bias} we introduce the unentangled bias $\bias$, the complex bias $\bias^\C$, and the relaxation $\biasnc$ and prove inequalities relating them (see Theorem~\ref{thm:unentangled-bias}). In Section~\ref{sec:entangled-bias} we introduce the entangled bias $\bias^*$, the maximally entangled bias $\biasme$, and the relaxation $\biasos$, and prove inequalities relating them (see Lemma~\ref{lem:bias-relaxation-quantum} and Theorem~\ref{thm:entangled-bias}). 

\subsection{Definitions}\label{sec:quantum-defs}

We first give the mathematical definition of quantum XOR games and of strategies that we will be working with throughout the paper. After stating the definition we discuss different possible operational interpretations of quantum XOR games, all of which are captured by our definition. 

\begin{definition}\label{def:qxor}
A quantum XOR game $G$ of size $n$ is specified by a Hermitian matrix $M=M(G)\in \herm{\C^n\otimes \C^n}$ such that $\|M\|_1\leq 1$. A strategy for the players in $G$ is given by a pair of observables $A\in \obs{\C^n\otimes \HA}$, $B\in \obs{\C^n\otimes \HB}$, where $\HA,\HB$ are finite-dimensional Hilbert spaces, and a state $\ket{\Psi} \in \HA\otimes \HB$. The bias achieved by the strategy $(A,B,\ket{\Psi})$ in $G$ is 
\begin{align}
\bias(A,B,\ket{\Psi};G)\,&:=\,  \bra{\Psi} \Tr_{\C^n\otimes \C^n}\big((A\otimes B) \, (M\otimes \Id_{\HA\otimes\HB}) \big) \ket{\Psi} \notag \\[2mm]
&\phantom{:}= \, \Tr\big( (A\otimes B)\, (M\otimes \ket{\Psi}\bra{\Psi} )\big). \label{eq:def-biasm}
\end{align}
\end{definition}

We first observe that with this definition quantum XOR games are clearly a generalization of classical XOR games: if $G$ is a classical XOR game of size $n$ with coefficients $(R_{s,t})$, then one can obtain an equivalent quantum XOR game $G'$ by introducing the $n^2$-dimensional diagonal matrix $M = \sum_{s,t} R_{s,t} \ket{s}\bra{s}\otimes \ket{t}\bra{t}$, which satisfies $\|M\|_1 = \sum_{s,t} |R_{s,t}|=1$. Moreover, it is not hard to check that, given any strategy $((A_s),(B_t),\ket{\Psi})$ for the players in $G$, its bias equals $\bias(A,B,\ket{\Psi};G')$, where $A$ (resp.\ $B$) is the block-diagonal matrix with blocks the $A_s$ (resp.\ $B_t$). Conversely, any strategy $(A,B,\ket{\Psi})$ in $G'$ can be mapped to a strategy for the players in $G$ achieving the same bias by letting $A_s$ (resp. $B_t$) be the diagonal blocks of $A$ (resp. $B$), which are Hermitian of norm at most~$1$. 

\paragraph{Operational interpretations.}
Consider the actions of an arbitrary referee. First, he initializes the message registers and his private register, described by some Hilbert space $\mathcal{V}$, in an arbitrary state, which we can assume without loss of generality to be a pure state $\ket{\Phi_{init}} \in \C^n \otimes \C^n \otimes \mathcal{V}$. He then sends each message register to the corresponding player. 
The players apply arbitrary observables $A=A^0-A^1$, $B=B^0-B^1$ on their message and their own private spaces, initialized in an arbitrary state $\ket{\Psi}$. They return the outcomes $a,b$ of their measurements to the referee, who then measures his private register using either the binary measurement $\{\Pi_{0}^{acc},\Id-\Pi_{0}^{acc}\}$ or $\{\Pi_{1}^{acc},\Id-\Pi_{1}^{acc}\}$, depending on the parity $a\oplus b$. If he obtains the outcome ``$acc$'' he accepts, and otherwise he rejects. The success probability of the strategy $(A,B,\ket{\Psi})$ is
\begin{align}
&\bra{\Psi}\bra{\Phi_{init}} ((A^0\otimes B^{0}+ A^{1} \otimes B^{1}) \otimes \Pi_0^{acc} +
(A^0\otimes B^{1}+ A^{1} \otimes B^{0}) \otimes \Pi_1^{acc})\ket{\Psi} \ket{\Phi_{init}}  \notag \\
&\qquad = \frac{1}{2} 
\bra{\Phi_{init}} \Id_{\C^n \otimes \C^n}\otimes(\Pi_0^{acc} + \Pi_1^{acc}) \ket{\Phi_{init}} +
  \frac{1}{2} 
 \bra{\Psi} \bra{\Phi_{init}} (A \otimes B) \otimes (\Pi_0^{acc} - \Pi_1^{acc})  \ket{\Psi}\ket{\Phi_{init}} \notag  \\
&\qquad = \frac{1}{2} 
\bra{\Phi_{init}} \Id_{\C^n \otimes \C^n}\otimes(\Pi_0^{acc} + \Pi_1^{acc}) \ket{\Phi_{init}} +
\frac{1}{2} \Tr\big( (A\otimes B)\, (M\otimes \ket{\Psi}\bra{\Psi} )\big),\label{eq:operationaldef}
\end{align}
where we define
$$ M \,:=\, \Tr_{\mathcal{V}}\big( (\Id_{\C^n\otimes \C^n} \otimes  (\Pi_0^{acc}-\Pi_1^{acc})) \ket{\Phi_{init}} \bra{\Phi_{init}}\big).$$
Notice that if the players output random uniform bits, then their success probability is given by the first term in~\eqref{eq:operationaldef}, and therefore, the bias 
$\bias(A,B,\ket{\Psi};G)$ as defined in~\eqref{eq:def-biasm} corresponds exactly to twice the advantage of players using the strategy $(A,B,\ket{\Psi})$ in $G$ over players applying the random strategy. 

\begin{topbotframe}
\begin{example}[Matrix associated to the family of games $(T_n)$]\label{ex:pt-m} The matrix $M$ associated to the game $T_n$, defined in Section~\ref{sec:ourresults}, is 
\begin{align*}
M({T_n}) \,&=\, \frac{1}{2}\big(\ket{\psi_0}\bra{\psi_0} - \ket{\psi_1}\bra{\psi_1}\big) \\
\,&=\,  \frac{1}{2\sqrt{n}} \Big(\sum_{i=1}^n \ket{00}\bra{ii} + \ket{ii}\bra{00}\Big) .
\end{align*}
\end{example}
\end{topbotframe}

Conversely, we show that to any Hermitian $M$ satisfying $\|M\|_1\leq 1$ may be associated a quantum XOR game in which the players' bias is given by~\eqref{eq:def-biasm}. Indeed, for any such $M$ we may write the spectral decomposition $M = \sum_i (-1)^{c_i} p_i \ket{\Phi_i}\bra{\Phi_i}$, where the $p_i$ are non-negative and sum to $\|M\|_1$. It is then easy to check that $M$ is associated to the following game by the transformation described above. The referee first selects an $i\in [n^2]$ with probability $p_i$, and rejects outright with probability $1-\sum_i p_i$. Provided this last option did not happen, he prepares the $n^2$-dimensional state $\ket{\Phi_i}\in \C^n\otimes \C^n$ corresponding to the index $i$ he obtained, and sends one register of $\ket{\Phi_i}$ to each player. The referee accepts answers $(a,b)$ if and only if $a\oplus b = c_i$. Note that the states sent by the referee in this game are all orthogonal, hence can be perfectly distinguished globally. Alternatively, we could also decompose $M$ as $M = p_0\rho_0 - p_1 \rho_1$ with $p_0+p_1 = \|M\|_1$, and then have the referee send one of two possible density matrices to the players.

Up to a multiplicative scaling of the bias, one may even turn any $M$ into a quantum XOR game $G$ in which the referee always sends (not necessarily orthogonal) \emph{product} states to the players. To see this, let $\{H_i\}\in\herm{\C^n}$ be a basis of the space of $n$-dimensional Hermitian matrices normalized to have $\|H_i\|_1=1$ for each $i$. Decompose $M$ in the tensor product basis $\{H_i\otimes H_j\}$ as $M = \sum_{i,j} m_{i,j}\, H_i\otimes H_j$ for some $m_{i,j} \in \R$. By applying an appropriate scaling (which will affect the bias correspondingly) we may reduce to the case in which $\|M\|_1\leq \sum_{i,j} |m_{i,j}|=1$. The resulting game can be described as follows. The referee first selects a pair of indices $(i,j)$ according to the distribution $\{|m_{i,j}|\}$, and then plays the quantum XOR game described by $\text{sign}(m_{i,j}) H_i\otimes H_j$, whose eigenvectors are all product states.  

\subsection{The unentangled bias}\label{sec:unentangled-bias}

In this section we introduce the \emph{unentangled bias} $\bias(G):=\bias(M(G))$ of a quantum XOR game $G$ of size $n$, which is the maximum bias achievable by players who do not have any shared entanglement. Formally, by Definition~\ref{def:qxor} specialized to the case of empty private spaces $\HA,\HB$, we obtain the following.

\begin{definition}\label{def:bias-unentangled} Let $n$ be an integer and $M\in\herm{\C^n\otimes \C^n}$. The \emph{unentangled bias} of $M$, denoted $\bias(M)$, is defined as
\beq\label{eq:omegau-def}
\bias(M) \,:=\,\sup_{ \substack{A\in \herm{\C^n},\,B\in \herm{\C^n}\\ \|A\|_\infty, \|B\|_\infty\leq 1}}\, \big|\Tr\big((A\otimes B) \, M \big)\big|.
\eeq
\end{definition}

The supremum in~\eqref{eq:omegau-def} is taken over all Hermitian operators $A,B$ acting directly on the players' respective message spaces. We note that by linearity, the supremum will always be achieved by $A,B$ which have all their eigenvalues in $\{\pm 1\}$, i.e., observables.

One might argue that the above definition is too strict, and we should allow the players to have their own private auxiliary space, initialized in the state $\ket{0}$. The following claim shows that this does not affect the definition of the unentangled bias. 

\begin{claim}\label{claim:obs_value-herm} Let $M \in \herm{\C^n\otimes \C^n}$. Then 
\beq\label{eq:omegau-herm}
\bias(M) \,=\, \sup_{\substack{\HA,\HB,\, A\in \obs{\C^n\otimes\HA},\\ B\in \obs{\C^n\otimes\HB}}} \,\big|\bra{0}_{\HA} \bra{0}_{\HB} \Tr_{\C^n\otimes \C^n}\big((A\otimes B) \,( M \otimes \Id_{\HA\otimes\HB})\big) \ket{0}_{\HA}\ket{0}_{\HB}\big|,
\eeq
where the supremum is taken over all finite-dimensional Hilbert spaces $\HA,\HB$. 
\end{claim}

\begin{proof} The $\le$ direction is clear. For the other direction, consider for any $A\in \obs{\C^n\otimes\HA}$ the matrix $A' = (\Id \otimes \bra{0}_{\HA}) A (\Id \otimes \ket{0}_{\HA})$ and similarly for $B$. Then $A'$ and $B'$ are Hermitian with norm at most $1$, and achieve the same bias. 
\end{proof}

\begin{topbotframe}
\begin{example}[Unentangled bias of the games $(T_n)$ (1)]\label{ex:pt-unentangled-1} The maximum bias achievable in the game $T_n$ satisfies $\bias(T_n) \geq 1/\sqrt{n}$ (and in particular $\bias(T_1)=1$), as is demonstrated by the following strategy for the players. (In fact, the bias is exactly $1/\sqrt{n}$: we will prove a matching upper bound $\bias(T_n)\leq 1/\sqrt{n}$ in Example~\ref{ex:pt-unentangled-2} below.) The players each measure their respective message register in a basis containing the two orthogonal vectors 
$$\ket{\pi_0} = \frac{1}{\sqrt{2}}\ket{0} +\frac{1}{\sqrt{2n}} \sum_i\ket{i}\quad\text{and}\quad \ket{\pi_1} = \frac{1}{\sqrt{2}}\ket{0} -\frac{1}{\sqrt{2n}} \sum_i\ket{i}.$$
If they obtain the outcome $\ket{\pi_0}$ (resp.\ $\ket{\pi_1}$) then they answer $0$ (resp.\ $1$); otherwise they output a random bit. Let $Q = \ket{\pi_0}\bra{\pi_0}-\ket{\pi_1}\bra{\pi_1}$. The bias achieved by this strategy is
\begin{align*}
 \frac{1}{2} \left( \bra{\psi_0} Q\otimes Q \ket{\psi_0}  - \bra{\psi_1} Q \otimes Q \ket{\psi_1} \right)
&=  \frac{1}{\sqrt{n}} {\rm Re}\left(\sum_{i=1}^n \bra{00} Q \otimes Q \ket{ii}\right)\\
    &= \frac{1}{\sqrt{n}} {\rm Re}\left(\sum_{i=1}^n \bra{0}Q \ket{i}^2  \right)\\
    & =\frac{1}{\sqrt{n}} \left(\sum_{i=1}^n \left(\frac{1}{\sqrt{n}}\right)^2  \right)\, =\, \frac{1}{\sqrt{n}}.
\end{align*}
\end{example}
\end{topbotframe}

It will sometimes be convenient to relax the condition that the operators $A,B$ in~\eqref{eq:omegau-herm} are Hermitian, and allow them to be arbitrary norm-$1$ operators $A\in \lin{\C^n}, B\in \lin{\C^n}$. This is analogous to the relaxation of the bias into the complex bias that we already introduced in the case of classical XOR games in the previous section. Formally, we define the complex bias of any $M\in \lin{\C^n\otimes \C^n}$ as follows.\footnote{Although quantum XOR games only give rise to Hermitian matrices $M$, the quantities $\bias^\C$, as well as $\biasnc$ and $\biasos$ defined later, are meaningful for all $M$ and so we give their definitions in the general case.}

\begin{definition}\label{def:bias-complex} Let $n$ be an integer and $M\in\lin{C^n\otimes \C^n}$. The \emph{complex bias} of $M$, $\biasc(M)$, is defined as
\beq\label{eq:omegac-def}
 \biasc(M)\,:=\,\sup_{A,B\in\lin{\C^n},\,  \|A\|_{\infty}\leq 1,\,\|B\|_{\infty}\leq 1} \,\big| \Tr\big((A\otimes B) \, M \big)\big|.
\eeq
\end{definition}

The following claim shows that $\biasc(M)$ is never more than a factor $\sqrt{2}$ larger than the unentangled bias $\bias(M)$; the fact that such a gap can be achieved already follows from Example~\ref{ex:chsh-complex}. 

\begin{claim}\label{claim:obs_value} Let $M\in\herm{\C^n\otimes \C^n}$. Then it holds that
$$ %\label{eq:u_value-0}
\bias(M) \, \leq\, \biasc(M)\,\leq \sqrt{2}\,\bias(M).
$$
\end{claim}

\begin{proof}
The first inequality is clear. 
To prove the second, let $A,B\in\lin{\C^n}$ achieve the supremum in~\eqref{eq:omegac-def}. 
By convexity, we can assume without loss of generality that $A,B$ are extreme points of the set of all operators of norm at most $1$; hence all their singular values must be $1$, i.e., they are unitary. We may thus decompose
$$A = \sum_i \lambda_i \ket{u_i}\bra{u_i}\quad\text{and}\quad B = \sum_i \mu_i \ket{v_i}\bra{v_i},$$
where the $\lambda_i$ and $\mu_i$ are complex with modulus $1$. The complex bias is then
$$ \bias^\C(M)\,=\, \Big|\sum_{i,j} \big((\bra{u_i}\otimes\bra{v_j}) M (\ket{u_i}\otimes \ket{v_j})\big) \,\lambda_i \mu_j\Big|.$$
By multiplying all $\lambda_i$ by a complex phase we can assume that the expression inside the absolute value is a non-negative real.
For each $i$ define the two-dimensional real unit vectors $\vec{\lambda}_i = (\Re(\lambda_i),\Im(\lambda_i))^T$ and $\vec{\mu}_i = (\Re(\mu_i),-\Im(\mu_i))^T$,
and notice that $\langle \vec{\lambda}_i, \vec{\mu}_j \rangle = \Re(\lambda_i \mu_j)$. 
Since $M$ is Hermitian, for every $i,j$ the coefficient $(\bra{u_i}\otimes\bra{v_j}) M (\ket{u_i}\otimes \ket{v_j})$ is real, and so we have
$$ \bias^\C(M)\,=\, \sum_{i,j} \big((\bra{u_i}\otimes\bra{v_j}) M (\ket{u_i}\otimes \ket{v_j})\big) \,\langle \vec{\lambda_i}, \vec{ \mu_j} \rangle.$$
Using Krivine's~\cite{Krivine:79a} result that the two-dimensional Grothendieck constant $K_G^\R(2)$ is $\sqrt{2}$, 
we obtain that there exist numbers $x_i,y_j \in \{-1,1\}$ such that 
$$ \sum_{i,j} \big((\bra{u_i}\otimes\bra{v_j}) M (\ket{u_i}\otimes \ket{v_j})\big) \, x_i y_j \ge \bias^\C(M) / \sqrt{2}.$$
We can therefore complete the proof by using in~\eqref{eq:omegau-def} the observables 
\[ A' \,:=\, \sum_i x_i \ket{u_i}\bra{u_i}\quad\text{and}\quad B' \,:=\, \sum_i y_i \ket{v_i}\bra{v_i}.\qedhere \]
\end{proof}

Next, we introduce the relaxation $\biasnc$. We start with some notation. 
Given a Hilbert space $\HA$ and an integer $d$ we denote by $\vecm{d}{\HA}$ the complex vector space of all sequences of $d$ matrices $A_1,\ldots,A_d\in \lin{\HA}$; given such a sequence we will use the notation $\vec{A} := (A_1,\ldots,A_d)$ to represent it. This notation emphasizes the fact that $\vec{A}$ can be thought of both as a sequence of matrices, or as the vector-valued matrix whose $(i,j)$th entry is the vector $((A_1)_{i,j},\ldots,(A_d)_{i,j})\in\C^d$. 
Given two vector-valued matrices $\vec{A} = (A_1,\ldots,A_d) \in\vecm{d}{\HA}$ and $\vec{B} = (B_1,\ldots,B_d)\in\vecm{d}{\HB}$, we define their tensor product ``$\odot$'' as the complex-valued matrix
\begin{align*}
\vec{A} \odot \vec{B} \,&:=\, \sum_{r=1}^d A_r \otimes B_r \\
&\phantom{:}=\, \big( \big\langle \overline{\vec{A}_{i,k}} , \vec{B}_{j,l} \big\rangle \big)_{(i,j),(k,l)}
     \,\in \lin{\HA\otimes \HB}.
\end{align*}
In other words, the tensor product of two vector-valued matrices is defined as that of scalar-valued matrices, except we take inner products of entries instead of scalar products.
We also define the product of two vector-valued matrices $\vec{A}\in\vecm{d}{\HA,\HB}$ and $\vec{B}\in\vecm{d}{\HB,\HC}$ as the complex-valued matrix
\begin{align*}
 (\vec{A} \vec{B})_{i,j} \,&:=\, \sum_{r=1}^d A_i B_i \\
 &\phantom{:}= \, \sum_k \big\langle \overline{\vec{A}_{i,k}} , \vec{B}_{k,j}\big\rangle \,\in \lin{\HA,\HC}.
\end{align*}
Note that $\vec{A}\vec{B}$ is obtained in the same way as the usual matrix product, except that we are taking the inner product, rather than the product, of corresponding entries. 
Finally, given $\vec{A} = (A_1,\ldots,A_d)\in\vecm{d}{\HA}$ we define its dagger as $\vec{A}^\dagger = (A_1^\dagger,\ldots,A_d^\dagger)$. 

\begin{definition}\label{def:bias-nc} Let $n$ be an integer and $M\in \lin{\C^n\otimes \C^n}$. Define
\begin{align}\label{eq:nc-gt-p}
\biasnc(M)\,:=\,\sup_{d,\,\vec{X},\vec{Y}\in\vecm{d}{\C^n}} \,\big|\Tr\big( (\vec{X}\odot \vec{Y})\, M\big)\big|,
\end{align}
where the supremum is taken over all dimensions $d$ and vector-valued matrices $\vec{X}\in\vecm{d}{\C^d}$ and $\vec{Y}\in\vecm{d}{\C^n}$ satisfying %\vskip-0.7cm
\begin{align}
\max\Big( \big\|\vec{X} \vec{X}^\dagger\big\|_\infty,\,\big\|\vec{X}^\dagger \vec{X}\big\|_\infty,&\,\big\|\vec{Y} \vec{Y}^\dagger\big\|_\infty,\,\big\|\vec{Y}^\dagger \vec{Y}\big\|_\infty\Big) \,\leq\, 1.\label{eq:nc-gt-pcons}
\end{align}
\end{definition}

If we restrict the supremum to $d=1$ then the constraint~\eqref{eq:nc-gt-pcons} simply expresses that $\vec{X} = (X)$ and $\vec{Y} = (Y)$ should have norm at most $1$, so that $\biasnc$ is indeed a relaxation of the complex bias, i.e., $\biasc(M)\leq \biasnc(M)$. Note moreover that if $M$ is a real diagonal matrix then only the vectors on the diagonal of $\vec{X}$ and $\vec{Y}$ contribute to~\eqref{eq:nc-gt-p}. The constraint~\eqref{eq:nc-gt-pcons} implies that these vectors must have norm at most $1$. Therefore, in the case of a classical XOR game we have $\biasnc(G)=\biascom(G)$, the relaxation of the bias defined in~\eqref{eq:class-bias-sdp}.

\bigskip

%\noindent\begin{minipage}{\textwidth}
\begin{topbotframe}
\begin{example}[Unentangled bias of the games $(T_n)$ (2)]\label{ex:pt-unentangled-2} We now observe that for the family $(T_n)$ we have $\biasnc(T_n)\leq 1/\sqrt{n}$, implying that $\biasnc(T_n)=\bias^\C(T_n)=\bias(T_n)=1/\sqrt{n}$. Indeed, let $\vec{X}\in\vecm{d}{\C^n}$ and $\vec{Y}\in\vecm{d}{\C^n}$ be arbitrary vector-valued matrices satisfying the constraints~\eqref{eq:nc-gt-pcons}. The resulting value in~\eqref{eq:nc-gt-p} is
\begin{align*}
 \Tr\big( (\vec{X}\odot \vec{Y})\, T_n\big) &= 
 \frac{1}{2\sqrt{n}} \Big(\sum_{i=1}^n \sum_{j=1}^d \big( \bra{00} X_j \otimes Y_j \ket{ii} + \bra{ii} X_j \otimes Y_j \ket{00}\big)\Big)\\
    &= \frac{1}{2\sqrt{n}} \Big(\Big(\sum_{i=1}^n \sum_{j=1}^d \bra{0} X_j\ket{i} \bra{0} Y_j \ket{i}\Big) + \Big( \sum_{i=1}^n \sum_{j=1}^d \bra{i} X_j\ket{0} \bra{i} Y_j \ket{0}\Big)\Big) \\
    &\le \frac{1}{2\sqrt{n}} \Big(\Big(\sum_{i=1}^n\sum_{j=1}^d |\bra{0}X_j \ket{i}|^2\Big)^{1/2} \Big(\sum_{i=1}^n\sum_{j=1}^d |\bra{0}Y_j \ket{i}|^2\Big)^{1/2}\\
		&\hskip2cm +\Big(\sum_{i=1}^n\sum_{j=1}^d |\bra{i}X_j \ket{0}|^2\Big)^{1/2} \Big(\sum_{i=1}^n\sum_{j=1}^d |\bra{i}Y_j \ket{0}|^2\Big)^{1/2}\Big) \\
		&\leq \frac{1}{2\sqrt{n}} \Big( \Big\| \sum_j  X_j X_j^\dagger\Big\|_\infty^{1/2}\Big\| \sum_j  Y_j Y_j^\dagger\Big\|_\infty^{1/2} + \Big\| \sum_j   X_j^\dagger X_j\Big\|_\infty^{1/2}\Big\| \sum_j  Y_j^\dagger Y_j \Big\|_\infty^{1/2}\Big)\\
&		\le \frac{1}{\sqrt{n}},
\end{align*}
where the first inequality follows from the Cauchy-Schwarz inequality, the second uses that for any $Z$,
$$\sum_{i=1}^n |\bra{0}Z \ket{i}|^2 \,\leq\, \bra{0} ZZ^\dagger \ket{0}
\qquad\text{and}\qquad 
\sum_{i=1}^n \,|\bra{i}Z \ket{0}|^2 \,\leq\, \bra{0} Z^\dagger Z\ket{0},$$
and the last follows since $\vec{X},\vec{Y}$ satisfy~\eqref{eq:nc-gt-pcons}.  
\end{example}
\end{topbotframe}
%\end{minipage}

The following lemma shows that $\biasnc(M)$ is never ``unreasonably large", that is, it is never larger than $\|M\|_1$, which is the bias that the players would obtain if they were allowed to apply a \emph{single} joint unitary simultaneously on both their message registers (something one might call
the ``colluding bias"). 

\begin{lemma}\label{lem:realval} Let $n$ be an integer and $M\in \lin{\C^n\otimes\C^n}$. Then
$$ \biasnc(M) \,\leq\, \|M\|_1.$$
\end{lemma}

\begin{proof}
Let $\vec{X},\vec{Y}$ be vector-valued matrices satisfying the constraint~\eqref{eq:nc-gt-pcons}. Write $\vec{X} = (X_i)$ (resp. $\vec{Y}=(Y_i)$), where each $X_i$ (resp. $Y_i$) is in $\lin{\C^n}$, and let $M = \sum_j s_j \ket{u_j}\bra{v_j}$ be the singular value decomposition of $M$. Then by the Cauchy-Schwarz inequality, 
\begin{align*}
| \Tr((\vec{X}\odot \vec{Y}) M)| &\leq \sum_{i,j} s_j\, \big| \bra{v_j} X_i\otimes Y_i \ket{u_j} \big|\\
&\leq  \sum_{i,j} s_j\, \Big(\bra{v_j} X_iX_i^\dagger \otimes \Id \ket{v_j}\Big)^{1/2} \Big(\bra{u_j} \Id \otimes Y_i^\dagger Y_i \ket{u_j}\Big)^{1/2}\\
&\leq  \sum_{j} s_j \, \Big(\sum_i \bra{v_j} X_iX_i^\dagger \otimes \Id \ket{v_j}\Big)^{1/2}\Big(\sum_i \bra{u_j} \Id \otimes Y_i^\dagger Y_i \ket{u_j}\Big)^{1/2}\\
&\leq \sum_{j} s_j = \|M\|_1,
\end{align*}
where the last inequality follows from the constraint~\eqref{eq:nc-gt-pcons}.
\end{proof}

The proof of Lemma~\ref{lem:realval} only makes use of two of the four constraints in~\eqref{eq:nc-gt-pcons}, and it holds as long as either both constraints $\|\vec{X} \vec{X}^\dagger\|_\infty$, $\|\vec{Y}^\dagger \vec{Y}\|_\infty \leq 1$, or both constraints $\| \vec{X}^\dagger\vec{X}\|_\infty,\| \vec{Y}\vec{Y}^\dagger\|_\infty \leq 1$, hold. If instead one was to keep only (say) the two constraints  $\|\vec{X} \vec{X}^\dagger\|_\infty,\| \vec{Y}\vec{Y}^\dagger\|_\infty \leq 1$, then the lemma would no longer be true. This can be seen by taking $M$ to be the matrix associated with the game $T_n$, as in Example~\ref{ex:pt-m}.
Let $\vec{X} = \vec{Y} \in \vecm{n}{\C^{n+1}}$ have the basis vector $e_i\in\R^n$ in position $(i,0)$ for $i=1,\ldots,n$, and $0$ elsewhere. Then $\vec{X}\vec{X}^\dagger=\vec{Y}\vec{Y}^\dagger=\mathrm{diag}(0,1,\ldots,1)$ so both constraints are satisfied. However, one can easily compute
$$ \Tr\big( (\vec{X}\odot \vec{Y}) M\big) \,=\, \frac{1}{2\sqrt{n}} \sum_i 1 \,=\, \frac{\sqrt{n}}{2},$$
which is much larger than $\|M\|_1=1$. 

In addition, we note that all four constraints in~\eqref{eq:nc-gt-pcons} are necessary in order for $\biasnc(M)$ to be a constant-factor relaxation of $\bias(M)$. Indeed, suppose for example that we drop the constraint $\| \vec{X}^\dagger\vec{X}\|_\infty\leq 1$. Then $\vec{X}$ as defined above, and $\vec{Y} = \vec{X}/\sqrt{n} $ would constitute a feasible solution, with corresponding value  $\Tr( (\vec{X}\odot \vec{Y}) M) = 1/2$: this is much larger than $\bias(M) = 1/\sqrt{n}$. The following theorem states that, when all constraints are present, $\biasnc$ indeed gives a constant factor approximation to both the unentangled and complex biases. 

\begin{theorem}\label{thm:unentangled-bias} Let $G$ be a quantum XOR game of size $n$. For any $\eps>0$, one can approximate up to $(1\pm\eps)$ in time $\poly(n,\log 1/\eps)$ a quantity $\biasnc(G)$ which satisfies
$$ \bias(G) \,\leq\,\bias^\C(G)\,\leq\,\biasnc(G) \,\leq\, 2\,\biasc(G)\leq 2\sqrt{2}\,\bias(G).$$
Moreover, there is an infinite sequence of games for which the ratio $\biasnc(G)/\bias^\C(G)$ converges to $2$. Also, it always holds that $\biasnc(G) \leq \|M(G)\|_1 \leq 1$. Finally, the upper bounds on $\biasnc$ are explicit, in the sense that the algorithm can also output, in time $\poly(n,1/\eps)$, a description of a complex strategy achieving bias $\biasnc(G)/(2+\eps)$ as well as a strategy achieving bias $\biasnc(G)/(2\sqrt{2}+\eps)$. 
\end{theorem}

\begin{proof}
It is not difficult to verify that $\biasnc$, just like $\biascom$, is a semidefinite program. As such, it can be solved up to precision $\eps$ in time $\poly(n,\log 1/\eps)$. In slightly more detail, the semidefinite program corresponding to $\biasnc$ is over $2 n^2$ vector variables and the goal function is a linear function in the inner products between these vectors. To see why the constraint~\eqref{eq:nc-gt-pcons} is a semidefinite constraint, it suffices to notice that (say) $\|\vec{X} \vec{X}^\dagger\|_\infty \le 1$ is equivalent to $\vec{X} \vec{X}^\dagger \le \Id$.
\pnote{we cannot really write both $X X^\dagger$ and $X^\dagger X$; I think instead of the latter we can only write $\overline{X^\dagger X}$, i.e., its complex conjugate, but maybe this need not be mentioned}

The first and last inequalities follow from Claim~\ref{claim:obs_value}. The second inequality was already observed above. 
The substance of the theorem is in the inequality $\biasnc(G) \leq 2\,\biasc(G)$. This inequality is a consequence of the ``non-commutative Grothendieck inequality'' proved by Pisier~\cite{Pisier78NCGT} and Haagerup~\cite{Haagerup85NCGT}. While technically it follows directly from that result, the connection may not be immediate to readers unfamiliar with the uses of Grothendieck's inequality made in the functional analysis literature, and we explain the derivation in detail in Section~\ref{sec:grothendieck}.

The gap $\biasnc/\bias^\C\to 2$ follows from Theorem~\ref{thm:higamesummary}.
The explicit forms of the upper bounds follow from the algorithmic variant of the non-commutative Grothendieck inequality, as detailed in~\cite{RegevV12c}.
\end{proof}

\subsection{The entangled bias}\label{sec:entangled-bias}

We now consider the case that the players are allowed to initialize their private spaces $\HA$, $\HB$ in an arbitrary state $\ket{\Psi}\in \HA\otimes \HB$. Following Definition~\ref{def:qxor}, the resulting \emph{entangled bias} $\bias^*(G):= \bias^*(M(G))$ can be defined as follows. 

\begin{definition}\label{def-bias} Let $n$ be an integer and $M\in\herm{\C^n\otimes \C^n}$. The \emph{entangled bias} of $M$, denoted $\bias^*(M)$, is defined as 
\beq\label{eq:omegaq-def}
\bias^*(M) :=\, \sup_{\substack{\HA,\HB,\ket{\Psi},\, A\in \obs{\C^n\otimes\HA},\\ B\in \obs{\C^n\otimes\HB}}} \big| \bra{\Psi} \Tr_{\C^n\otimes \C^n}\big((A\otimes B) \, (M\otimes \Id_{\HA\otimes \HB}) \big) \ket{\Psi} \big|,
\eeq
where the supremum is taken over all finite dimensional Hilbert spaces $\HA,\HB$ and states $\ket{\Psi}\in\HA\otimes\HB$.
\end{definition}

\begin{topbotframe}
\begin{example}[Entangled bias of the games $(T_n)$]\label{ex:pt-bias}
We show that for any $n$ the game $T_n$ can be won with probability arbitrarily close to $1$, provided the players are allowed to share an entangled state of large enough dimension. 
First recall from Example~\ref{ex:pt-unentangled-1} that $T_1$ can be won with probability $1$ (even without any entanglement). In order to succeed in the game $T_n$ for general $n$, the players will use a specific entangled state in order to reduce to the case $n=1$. For any $d\ge 1$ this state, which falls in the family of so-called~\emph{embezzlement states},\footnote{The specific state we use was introduced in~\cite{LTW08}. See also~\cite{vDH03} for a ``universal'' family of states having similar ``embezzlement'' properties.} is defined as  
$$ \ket{\emb_{d}} \,:=\, \frac{1}{\sqrt{d}}\sum_{j=1}^d \big(\ket{n+1}\ket{n+1}\big)^{\otimes j} \otimes \ket{\Psi_n^{me}}^{\otimes (d-j)} \,\in\, (\C^{n+1})^{\otimes d} \otimes (\C^{n+1})^{\otimes d}.$$
Consider the following strategy for the players in $T_n$, defined for any integer $d$. The players initialize their private registers in state $\ket{\emb_{d}}$. Upon receiving their respective message register, controlled on the message register not being in state $\ket{0}$ they each apply the unitary transformation corresponding to a cyclic shift on the $d+1$ copies of $\C^{n+1}$ in their possession. This leads to the transformation\footnote{Note that this transformation requires the players' message registers to be of dimension $n+2$ instead of $n+1$. This is easily achieved by having the players use an additional qubit as ancilla each.}
\begin{align*}
 \big(\ket{0}\ket{0} \big)\otimes \ket{\emb_{d}} \,&\mapsto\, \big(\ket{0}\ket{0} \big)\otimes \ket{\emb_{d}},\\
 \ket{\Psi_n^{me}}\otimes \ket{\emb_{d}} \,&\mapsto\, \ket{n+1}\ket{n+1}\otimes \Big( \frac{1}{\sqrt{d}}\sum_{j=0}^{d-1} \big(\ket{n+1}\ket{n+1}\big)^{\otimes j} \otimes \ket{\Psi_n^{me}}^{\otimes (d-j)}\Big).
\end{align*}
Since the state on the right has overlap $1-O(1/d)$ with $\ket{\emb_{d}}$, after the cyclic shift and up to a local unitary mapping $\ket{n+1}\mapsto\ket{1}$ the players' message registers are in a state close to what it would be in the game $T_1$. They may then apply their perfect strategy for $T_1$, in which case one can verify that they will succeed with probability $1-O(1/d)$ in $T_n$. We refer to the proof of Lemma~\ref{lem:rankone-equiv-rev} for more details.
\end{example}
\end{topbotframe}

In contrast to the unentangled case, relaxing the supremum in~\eqref{eq:omegaq-def} to be taken over all complex matrices with operator norm at most $1$ does not change the definition of the bias, as is shown in the following claim. 

\begin{claim}\label{claim:obs_value-quant} Let $M\in\herm{\C^n\otimes \C^n}$. Then
\beq\label{eq:obs_value-0}
\bias^*(M) \, =\, \sup_{\substack{\HA,\HB,\,\ket{\Phi},\ket{\Psi},\,A\in\lin{\C^n\otimes\HA},\\ B\in\lin{\C^n\otimes\HB},\, \|A\|_{\infty}\leq 1, \|B\|_{\infty}\leq 1}} \big|\bra{\Phi} \Tr_{\C^n\otimes \C^n}\big((A\otimes B) \, (M\otimes \Id_{\HA\otimes\HB}) \big) \ket{\Psi}\big|,
\eeq
where the supremum is taken over all finite-dimensional Hilbert spaces $\HA,\HB$ and states $\ket{\Phi},\ket{\Psi}\in\HA\otimes\HB$.
\end{claim}

\begin{proof}
It will suffice to show that the supremum in~\eqref{eq:omegaq-def} is at least as large as that in~\eqref{eq:obs_value-0}, since the other inequality is clear. Let $\HA,\HB$ be finite-dimensional Hilbert spaces, and let $\ket{\Phi},\ket{\Psi},A,B$ achieve the supremum in the right-hand side of~\eqref{eq:obs_value-0}. Without loss of generality we may assume that $A,B$ are unitary and the expression inside the absolute value
 is real and non-negative. Consider the two observables
$$\tilde{A} \,=\,\begin{pmatrix} 0 & A \\ A^\dagger & 0 \end{pmatrix} \in \obs{\C^n \otimes (\HA \otimes \C^2)} 
    \qquad \text{and} \qquad  
   \tilde{B} \,=\,\begin{pmatrix} 0 & B \\ B^\dagger & 0 \end{pmatrix} \in \obs{\C^n \otimes (\HB \otimes \C^2)},$$
and the state $\ket{\tilde{\Psi}} =  \frac{1}{\sqrt{2}} \big( \ket{\Phi} \otimes\ket{00}+\ket{\Psi}\otimes \ket{11}\big) \in 
  (\HA \otimes \C^2) \otimes (\HB \otimes \C^2)$. 
Then
\begin{align}
\bra{\tilde{\Psi}}\Tr_{\C^n\otimes \C^n}\big((\tilde{A}\otimes \tilde{B}) \, (M\otimes \Id_{\HA\otimes\HB}) \big)\ket{\tilde{\Psi}} &= \frac{1}{2}\,\big(\bra{\Phi}\Tr_{\C^n\otimes \C^n}\big((A\otimes B) \, (M\otimes \Id_{\HA\otimes\HB}) \big)\ket{{\Psi}} \label{eq:obs_value-1}\notag\\[2mm]
&\qquad + \bra{\Psi}\Tr_{\C^n\otimes \C^n}\big((A^\dagger \otimes B^\dagger) \, (M\otimes \Id_{\HA\otimes\HB}) \big)\ket{{\Phi}}\big)\notag\\
&= \bra{\Phi}\Tr_{\C^n\otimes \C^n}\big((A\otimes B) \, (M\otimes \Id_{\HA\otimes\HB}) \big)\ket{{\Psi}},\notag
\end{align}
since $M$ is Hermitian and given our assumption on the last expression above being real. 
\end{proof}

Next, we define the \emph{maximally entangled bias} $\biasme(G):=\biasme(M(G))$, in which players are restricted to sharing the maximally entangled state $\ket{\Psi^{me}}$. Following Definition~\ref{def:qxor}, it can be defined as follows.

\begin{definition}\label{def:bias-me} Let $n$ be an integer and $M\in\herm{\C^n\otimes \C^n}$. The \emph{maximally entangled bias} of $M$, denoted $\biasme(M)$, is defined as
\begin{equation}\label{eq:omegame-def}
\biasme(M)\,:=\, \sup_{d,\,A,B\in \obs{\C^n\otimes\C^d}}  \big| \bra{\Psi_d^{me}} \Tr_{\C^n\otimes \C^n}\big((A\otimes B) \, (M\otimes \Id_{\C^d\otimes \C^d}) \big) \ket{\Psi_d^{me}} \big|.
\end{equation}
\end{definition}

As in the proof of Claim~\ref{claim:obs_value-herm} it is easy to see that the absence of an explicit ancilla space for the players in~\eqref{eq:omegame-def} is without loss of generality.
We also note that as was the case for $\bias^*$, one can equivalently take the supremum here over all matrices with operator norm at most $1$. This follows by a straightforward modification of the proof of Claim~\ref{claim:obs_value-quant}. It also follows from this argument that for all games $G$, $\biasme(G) \ge \bias^\C(G)$; in fact, there exists a strategy using just one EPR pair that achieves bias $\bias^\C(G)$ in $G$.

The following lemma shows that $\biasnc(M)$ is always an upper bound on $\biasme(M)$. The lemma already appears in~\cite{HI95}, but we give a (slightly different) proof that will be useful to understand why $\biasnc(M)$ is in general \emph{not} an upper bound on the entangled bias $\bias^*(M)$ (as can be seen from the game $T_n$, which as shown in Examples~\ref{ex:pt-unentangled-2} and~\ref{ex:pt-bias} satisfies $\biasnc(T_n) = 1/\sqrt{n} \ll \bias^*(T_n)=1$). While reading the proof, the reader might wish to keep the proof of Lemma~\ref{lem:bias-relaxation-classical} in mind. 

\begin{lemma}\label{lem:bias-relaxation-quantum} Let $M\in\herm{\C^n\otimes \C^n}$. Then 
$$ \biasme(M) \,\leq\, \biasnc(M).$$
\end{lemma}

\begin{proof} Let $d$ be an integer and $A,B\in \obs{\C^n\otimes \C^d}$. The expression in the supremum in~\eqref{eq:omegame-def} can be written as
\begin{equation*} %\label{eq:quant-sdp-0}
 \frac{1}{d} \sum_{i,j} \Tr\big( (\Id \otimes \bra{i})A(\Id \otimes \ket{j})\otimes  (\Id \otimes \bra{i})B( \Id \otimes \ket{j}) M\big) =
 \Tr\big( (\vec{A}\odot\vec{B}) M \big)
\end{equation*}
where the vector-valued matrices $\vec{A},\vec{B}\in\vecm{d^2}{\C^n}$ are defined as
$$ \vec{A} \,:=\, \frac{1}{\sqrt{d}}\big( (\Id \otimes \bra{i})A(\Id \otimes \ket{j}) \big)_{i,j} \quad\text{and}\quad \vec{B}\,:=\, \frac{1}{\sqrt{d}}\big( (\Id \otimes \bra{i})B(\Id \otimes \ket{j}) \big)_{i,j}.$$
Note that if we think of $A$ as a $d \times d$ block matrix with each block of size $n \times n$, then 
$\vec{A}$ simply corresponds to a list of the $d^2$ blocks of $A$, and similarly for $\vec{B}$. In particular the matrix $\vec{A} \vec{A}^\dagger $ (resp.\ $\vec{A}^\dagger \vec{A}$) corresponds to the average of the $d$ diagonal blocks of the matrix $AA^\dagger$ (resp.\ $A^\dagger A$), and must therefore have operator norm at most $1$. Similar bounds hold for $\vec{B}\vec{B}^\dagger$ and $\vec{B}^\dagger \vec{B}$, showing that $(\vec{A},\vec{B})$ satisfy the constraint~\eqref{eq:nc-gt-pcons}, and the proof is complete. 
\end{proof}

As noted above, $\biasnc(G)$ is in general not an upper bound on the entangled bias $\bias^*(G)$. It is instructive to see what fails in the proof of Lemma~\ref{lem:bias-relaxation-quantum} if we try to 
adapt it to the case of a general entangled state. Following the proof of Lemma~\ref{lem:bias-relaxation-classical} we would have to weigh the block-rows of $A$ using the Schmidt coefficients of $\ket{\Psi}$, resulting in the vector-valued matrix
$$ \vec{A}_R \,=\, \big( \lambda_i\, (\Id \otimes \bra{i})A(\Id \otimes \ket{j}) \big)_{i,j},$$
and similarly define $\vec{B}_C$ using column-weighing.
Then $\Tr\big( (\vec{A}_R \odot\vec{B}_C) M \big)$ equals the expression inside the absolute value in~\eqref{eq:omegaq-def}, as desired.
Moreover, as before, $\vec{A}_R \vec{A}_R^\dagger$ is an average of the diagonal blocks of $AA^\dagger$, this time a weighted average with weights $\lambda_i$. As a result, we still have $\vec{A}_R\vec{A}_R^\dagger \leq \Id$. However, and this is where the proof fails, $\vec{A}_R^\dagger \vec{A}_R\leq \Id$ is no longer true in general, as is demonstrated by the following example. 

\begin{topbotframe}\pnote{Motivation for this example is that $\vec{X}_R = \vec{Y}_R = (\ket{0}\bra{i}/\sqrt{n}+\ket{i}\bra{0})_i$, $\vec{X}_C = \vec{Y}_C = (\ket{0}\bra{i}+\ket{i}\bra{0}/\sqrt{n})_i$ is a sdp solution showing $\biasos(T_n)=1$.}
\begin{example} Let $X$ be the matrix defined as $X = \sum_{i=1}^n \ket{1}\bra{i}\otimes \ket{i}\bra{0} + \sum_{i=1}^n \ket{i}\bra{1}\otimes \ket{0}\bra{i} \in \lin{\C^n \otimes \C^{n+1}}$. Then $X$ is Hermitian and of operator norm at most $1$, so $X$ is a valid strategy for either player in a quantum XOR game. Let $\ket{\Psi} = \frac{1}{\sqrt{2}}\ket{00} + \frac{1}{\sqrt{2n}} \sum_i \ket{ii}$. The corresponding ``row-weighted'' vector-valued matrix $\vec{X}_R = (X_{i,j})$ has most its entries equal to $0$, except for $X_{0,j} =  (1/\sqrt{2})\ket{j}\bra{1}$ for $j=1,\ldots,n$ and $X_{i,0} =  (1/\sqrt{2n})\ket{1}\bra{i}$ for $i=1,\ldots,n$. 
$\vec{X}_R$ satisfies $\vec{X}_R\vec{X}_R^\dagger =  (1/2)\ket{1}\bra{1}+(1/2)\Id \leq \Id$, but $\vec{X}_R^\dagger \vec{X}_R = (1/(2n)) \Id + (n/2)\ket{1}\bra{1}$ has operator norm $n/2+1/(2n)$. 
\end{example} 
\end{topbotframe}

One could, of course, weigh the columns of $A$ instead of its rows. 
This leads to the vector-valued matrix 
$$ \vec{A}_C \,=\, \big( \lambda_j\, (\Id \otimes \bra{i})A(\Id \otimes  \ket{j}) \big)_{i,j} ,$$
and to the similarly defined $\vec{B}_R$.
Now we have $\vec{A}_C^\dagger \vec{A}_C\leq \Id$ but in general not $\vec{A}_C \vec{A}_C^\dagger \leq \Id$! 

The discussion above explains why $\biasnc$ is not an upper bound on $\bias^*$. But not all is lost, and it turns out
that one can relax the constraint in the definition of $\biasnc$, leading to the quantity we call $\biasos$, which \emph{does} upper bound $\bias^*$. 
The idea is to include \emph{both} vector-valued matrices $\vec{A}_R$, $\vec{A}_C$, each satisfying the corresponding constraint, as well as matching $\vec{B}_C,\vec{B}_R$, and a consistency constraint among them. 

\begin{definition}\label{def:bias-os} Let $n$ be an integer and $M\in\lin{\C^n\otimes\C^n}$. Define
\begin{align}\label{eq:os-gt-p}
 \biasos(M)\,:=\,\sup_{d,\,\vec{X}_R,\vec{X}_C,\vec{Y}_R,\vec{Y}_C\in\vecm{d}{\C^n}} \big|\Tr\big( (\vec{X}_R\odot \vec{Y}_C)\, M\big)\big|,
\end{align}
where the supremum is taken over all vector-valued matrices $\vec{X}_R,\vec{X}_C,\vec{Y}_R,\vec{Y}_C \in \vecm{d}{\C^n}$ such that the following constraints hold:
\begin{align}
\vec{X}_R\odot \vec{Y}_C &\,=\, \vec{X}_C\odot \vec{Y}_R,\label{eq:os-gt-pcons}\\
\max\Big( \big\|\vec{X}_R \vec{X}_R^\dagger\big\|_\infty,\,\big\|\vec{Y}_R \vec{Y}_R^\dagger\big\|_\infty,&\,\big\|\vec{X}_C^\dagger \vec{X}_C\big\|_\infty,\,\big\|\vec{Y}_C^\dagger \vec{Y}_C\big\|_\infty\Big) \,\leq\, 1.\notag
\end{align}
\end{definition}

As we already saw in the setting of $\biasnc(G)$, in case $G$ is a classical XOR game  only the vectors appearing on the diagonal of $\vec{X}_R,\vec{X}_C,\vec{Y}_R$ and $\vec{Y_C}$ contribute to the objective value~\eqref{eq:os-gt-p}, and the constraints $\|\vec{X}_R \vec{X}_R^\dagger\|_\infty \leq 1$ and $\|\vec{Y}_C^\dagger \vec{Y}_C\|_\infty \le 1$ impose that these vectors have norm at most $1$. Hence in that case it holds that $\bias^*(G) = \biascom(G) = \biasnc(G) = \biasos(G)$ (the first equality was already shown in Proposition~\ref{prop:tsirelson}).

We end this section with a proof of the second sequence of inequalities in Theorem~\ref{thm:algorithm}, together with additional properties. 

\begin{theorem}\label{thm:entangled-bias} Let $G$ be a quantum XOR game of size $n$. For any $\eps>0$, one can approximate up to $(1\pm\eps)$ in time $\poly(n,\log 1/\eps)$ a quantity $\biasos(G)$ which satisfies
$$ \bias^*(G) \,\leq\, \biasos(G)\,\leq\, 2\,\bias^*(G).$$
Moreover, there exists a game for which both inequalities are strict. It also holds that $\biasos(G)\leq \|M(G)\|_1 \leq 1$.
Finally, the upper bound on $\biasos$ is explicit, in the sense that the algorithm can also output in time $(n/\eps)^{O(1/\eps)}$ a description
of an entangled strategy achieving bias $\biasos(G)/(2+\eps)$. The strategy uses entanglement of dimension at most $(n/\eps)^{C/\eps}$ for some universal constant $C>0$.
\end{theorem}

\begin{proof}
As was the case for $\biasnc$, it is not hard to see that $\biasos$ can be expressed as a semidefinite program  and is therefore computable up to precision $\eps$ in time $\poly(n,\log 1/\eps)$. The fact that $\bias^*(G)\leq\biasos(G)$ follows from the discussion preceding Definition~\ref{def:bias-os}. 
The substance of the theorem is in the inequality $\biasos(G)\leq 2\bias^*(G)$. The inequality follows from the operator space Grothendieck inequality; we state that inequality and explain in detail how it implies the inequality on the biases in Section~\ref{sec:grothendieck}. 

Theorem~\ref{thm:higamesummary} shows that $\bias^*(\HI_1)<\biasos(\HI_1)<2 \, \bias^*(\HI_1$),
where the games $(\HI_n)$ were introduced in Section~\ref{sec:ourresults}.
For the inequality $\biasos(M) \leq \|M\|_1$, see Lemma~\ref{lem:realval} and the remark following it. 

Finally, the explicit part of the theorem and the bound on the entanglement dimension follow from the quantitative version of the operator space Grothendieck inequality~\cite{RegevV12b}. 
\end{proof}

%% file: sec_rankone.tex
Rank-one quantum games are a model of two-player one-round games that was introduced recently by Cooney et al.~\cite{cooney}. Despite seemingly quite different from quantum XOR games (in particular, a rank-one game involves quantum communication both from the referee to the players \emph{and} from the players to the referee), in this section we show that there is a strong relationship between the two models in case players are allowed to use entanglement. 
In particular, there is a simple transformation from one type of game to the other that essentially maps the maximum success probability of entangled players in a rank-one game to the entangled bias of the corresponding quantum XOR game, and vice-versa. This equivalence may be a source of additional examples of quantum XOR games.  In fact, at the end of this section we show how the family of games $(T_n)$ can be naturally obtained from a simple rank-one game; this correspondence will let us rederive the fact, proved in Example~\ref{ex:pt-bias}, that $\bias^*(T_n)=1$. The family $(C_n)$ discussed in Section~\ref{sec:ourresults} is another interesting example. We note that Cooney et al.\ do not study either unentangled or maximally entangled players in their model; although one could define both, it is not clear how to analyze those quantities, and in particular they do not seem related to the non-commutative Grothendieck inequality, as is the case in our model. 
\pnote{The natural definition of the unentangled value in their model should allow the use of ancilla for the players. This allows them to win with good probability even on ``difficult" games like $T_n$. (The players can just use their ancilla to swap out entanglement.) Hence their unentangled value is not related to the norm of the bilinear form and the ncgt. The norm of the operator seem to correspond in their model to the case of \emph{no ancilla} which is somewhat unnatural.}

Formally, a rank-one game $\hat{G}$ of size $n$ is specified by an arbitrary (finite-dimensional) Hilbert space $\mathcal{V}$, corresponding to the referee's private space, and two unit vectors $\ket{\eta},\ket{\gamma}\in \C^n\otimes \C^n \otimes \mathcal{V}$. The game proceeds as follows. The referee first prepares the state $\ket{\eta}$ on three registers $\textsc{M}_A,\textsc{M}_B,\textsc{V}$ corresponding to the spaces $\C^n,\C^n,\mathcal{V}$ respectively. He sends register $\textsc{M}_A$ to Alice, and register $\textsc{M}_B$ to Bob. The players are allowed to apply arbitrary unitaries $U,V$ on their respective message registers as well as on their own private spaces $\HA$, $\HB$, which may be initialized in an arbitrary state $\ket{\Psi} \in \HA \otimes \HB$. The players then send registers $\textsc{M}_A$, $\textsc{M}_B$ back to the referee, who performs a \emph{rank-one} measurement $\big\{P^{acc} = \ket{\gamma}\bra{\gamma}, \Id - P^{acc}\big\}$ on the three registers in his possession. If he obtains the outcome ``\emph{acc}'' then he accepts, otherwise he rejects. 

Given a rank-one quantum game $\hat{G} = (\ket{\eta},\ket{\gamma})$, associate to it the (not necessarily Hermitian) matrix $\hat{M} = \hat{M}(\hat{G}) := \Tr_{\mathcal{V}} \ket{\eta}\bra{\gamma}$. It is not difficult to see that the maximum acceptance probability of any entangled players in $\hat{G}$ equals (see Theorem~3.2 in~\cite{cooney} for a proof)
\beq\label{eq:val-rank1}
\biasrk(\hat{G}) \,=\, \biasrk(\hat{M}(\hat{G}))\,:=\,\sup_{\substack{\HA,\HB,\,\ket{\Psi},\ket{\Phi}\in \HA\otimes\HB,\\U \in \lin{\C^n\otimes\HA},V\in \lin{\C^n\otimes\HB}}}\, \big|\Tr\big( (U\otimes V) (\hat{M}\otimes \ket{\Psi}\bra{\Phi})\big) \big|^2,
\eeq
 where the supremum is taken over all finite-dimensional Hilbert spaces $\HA$ and $\HB$, unit vectors $\ket{\Psi},\ket{\Phi} \in\HA\otimes \HB$, and $U\in\lin{\C^n\otimes \HA}$, $V\in\lin{\C^n\otimes\HB}$ of norm at most $1$. (As in the proof of Claim~\ref{claim:obs_value}, a convexity argument shows that we could equivalently restrict to unitaries.) We also include for completeness a proof of the following easy fact.

\begin{claim}[{\cite[Proposition~3.1]{cooney}}]\label{clm:gameformatrix}
For any matrix $\hat{M} \in \lin{\C^n\otimes \C^n}$ with $\|\hat{M}\|_1 \le 1$, there exists a rank-one game whose associated matrix is $\hat{M}$.
\end{claim}
\begin{proof}
Assume $\|\hat{M}\|_1 = 1$; the general case is similar.
Write $\hat{M}$ in its singular value decomposition as
$$\hat{M} \,=\,\sum_i \, s_i \, \ket{u_i}\bra{v_i},$$
with $s_i > 0$ and $\ket{u_i},\ket{v_i}$ orthonormal families in $\C^n\otimes \C^n$. Define a rank-one quantum game $\hat{G}$ by choosing the referee's private register $\mathcal{V}\simeq \C^{n^2}$, and letting 
$$\ket{\eta} \,:=\, \sum_i \sqrt{s_i} \ket{u_i}\ket{i}\quad\text{and}\quad\ket{\gamma} \,:=\, \sum_i \sqrt{s_i} \ket{v_i}\ket{i}.$$
It is easy to check that $\hat{M}(\hat{G}) = \hat{M}$.
\end{proof}

In the following two lemmas we state precisely the relationship between rank-one quantum games and quantum XOR games. 

\begin{lemma}\label{lem:rankone-equiv} Let $G$ be an arbitrary quantum XOR game. Then there exists a rank-one quantum game $\hat{G}$ such that $\biasrk(\hat{G}) = (\bias^*(G))^2$. Moreover, the associated matrices are equal: $\hat{M}(\hat{G}) = M(G)$.  
\end{lemma}

\begin{proof} 
Let $G$ be a quantum XOR game of size $n$ with associated game matrix $M$. Let $\hat{G}$ be a rank-one game with associated matrix $M$ as guaranteed to exist by  Claim~\ref{clm:gameformatrix}. Comparing Eq.~\eqref{eq:val-rank1} with the characterization of  $\bias^*(G)$ given in Claim~\ref{claim:obs_value-quant}, we see that $\biasrk(\hat{G}) = (\bias^*(G))^2$, as claimed.
\end{proof}

\begin{lemma}\label{lem:rankone-equiv-rev} Let $\hat{G}$ be an arbitrary rank-one quantum game. Then there exists a quantum XOR game $G$ such that $\bias^*(G) = (\biasrk(\hat{G}))^{1/2}$. 
\end{lemma}

\begin{proof}
 Let $\hat{G}=(\ket{\eta},\ket{\gamma})$ be a rank-one quantum game of size $n$. We associate to $\hat{G}$ the following quantum XOR game $G$ of size $2n$. In $G$, the referee prepares one of two possible states 
$$ \ket{\psi_\pm} \,:=\, \frac{1}{\sqrt{2}}\big(\ket{0}_{\textsc{M}_A}\ket{0}_{\textsc{M}_B} \ket{\eta}_{\textsc{M}_A\textsc{M}_B\textsc{V}} \pm \ket{1}_{\textsc{M}_A}\ket{1}_{\textsc{M}_B} \ket{\gamma}_{\textsc{M}_A\textsc{M}_B\textsc{V}}\big),$$
each with probability $1/2$. Note that the message registers in $G$ have one more qubit than those in $\hat{G}$. The referee sends the players their respective message registers, and accepts their answers if and only if their parity is $0$ in case the state prepared was $\ket{\psi_+}$, and $1$ in case it was $\ket{\psi_-}$. The corresponding game matrix is
\begin{align*}
 M \,=\, M(G) &= \frac{1}{2} \Big( \ket{00}\bra{11} \otimes \Tr_{\mathcal{V}}\big(\ket{\eta}\bra{\gamma}\big) + \ket{11}\bra{00}\otimes \Tr_{\mathcal{V}}\big(\ket{\gamma}\bra{\eta}\big)\Big)\\
&= \frac{1}{2} \Big( \ket{00}\bra{11} \otimes \hat{M} + \ket{11}\bra{00}\otimes \hat{M}^\dagger\Big),
\end{align*}
where $\hat{M} = \hat{M}(\hat{G})$ is the matrix associated to $\hat{G}$. 

We first show that $\bias^*(G) \leq (\biasrk(\hat{G}))^{1/2}$. Let $\eps>0$ and $(A,B,\ket{\Psi})$ a strategy for the players in $G$ achieving a bias at least $(1-\eps)\bias^*(G)$. Define $U := (\bra{1}\otimes \Id) A (\ket{0}\otimes \Id)$, $V := (\bra{1}\otimes \Id) B (\ket{0}\otimes \Id)$, and $\ket{\Phi} := \ket{\Psi}$. Then $U$ and $V$ each have norm at most $1$, so that the quadruple $(U,V,\ket{\Phi},\ket{\Psi})$ forms a valid assignment to the right-hand side of~\eqref{eq:val-rank1}. The resulting value is
\begin{align*}
\biasrk(\hat{G}) &\geq \big| \Tr\big( \big(U\otimes V\big) \big( \hat{M} \otimes \ket{\Psi} \bra{\Psi} \big)\big) \big|^2 \\
&= \big|\Tr\big( \big( A\otimes B\big) \big(M \otimes \ket{\Psi}\bra{\Psi}\big)\big|^2\\
&\geq ((1-\eps)\, \omega^*(G))^2,
\end{align*}
which concludes the proof of this direction of the inequality by letting $\eps\to 0$. 

It remains to show $\bias^*(G) \geq (\biasrk(\hat{G}))^{1/2}$. Let $U,V,\ket{\Psi}$ and $\ket{\Phi}$ achieve a value at least $(1-\eps)\biasrk(\hat{G})$ in the right-hand side of~\eqref{eq:val-rank1}. Assume without loss of generality that both $\ket{\Psi},\ket{\Phi} \in \HH \otimes \HH$ for some finite-dimensional $\HH$. Note that by changing the phase of $\ket{\Psi}$ we may assume that the expression inside the absolute value is real and non-negative. We construct a strategy for the players in $G$ based on the use of an embezzlement state.\footnote{A similar state was already used to prove $\bias^*(T_n)=1$ in Example~\ref{ex:pt-bias}.} Letting $d$ be a dimension parameter (we will eventually take the limit as $d\to\infty$), it is defined as 
$$ \ket{\emb_{d}} \,:=\, \frac{1}{\sqrt{D}}\sum_{j=1}^d \ket{\Psi}^{\otimes j} \otimes \ket{\Phi}^{\otimes (d-j)},$$
where $d\leq D \leq d^2$ is the appropriate normalization factor. 
The players share $\ket{\Phi}$ and $\ket{\emb_d}$, so each player 
has $d+1$ copies of $\HH$ altogether.
Let $\tilde{U},\tilde{V}$ be the unitary transformations corresponding to a cyclic shift on those $d+1$ copies, so that
$$\tilde{U}\otimes\tilde{V}:\, \ket{\Phi}\ket{\emb_d} \,\mapsto\, \ket{\Psi} \ket{\tilde{\emb}_d},$$
where $\ket{\tilde{\emb}_d} = (1/\sqrt{D})\sum_{j=0}^{d-1} \ket{\Psi}^{\otimes j} \otimes \ket{\Phi}^{\otimes (d-j)}$. Note that $\|\ket{\emb_d}-\ket{\tilde{\emb}_d}\|^2 \leq 4/d$ and so $\mathrm{Re}\langle \emb_d \ket{\tilde{\emb}_d} \ge 1-2/d$. Let
$$A \,:=\, \ket{0}\bra{1}\otimes (U\cdot \tilde{U}) + \ket{1}\bra{0} \otimes (U\cdot \tilde{U})^\dagger\quad\text{and}\quad B \,:=\, \ket{0}\bra{1}\otimes (V \cdot \tilde{V}) + \ket{1}\bra{0} \otimes (V \cdot \tilde{V})^\dagger.$$ 
One can verify that both $A$ and $B$ have norm at most $1$, and so by Claim~\ref{claim:obs_value-quant},
\begin{align*}
\bias^*(G) &\geq \big|\Tr\big( (A\otimes B)(M \otimes \ket{\Phi}\bra{\Phi} \otimes \ket{\emb_d}\bra{\emb_d})\big)\big| \\
&= \frac{1}{2}\Big|\bra{\eta}\bra{\Phi} \bra{\emb_d} U\tilde{U}\otimes V \tilde{V}\ket{\gamma}\ket{\Phi}\ket{\emb_d} + \bra{\gamma}\bra{\Phi} \bra{\emb_d} \tilde{U}^\dagger U^\dagger \otimes \tilde{V}^\dagger V^\dagger \ket{\eta}\ket{\Phi}\ket{\emb_d} \Big|\\
&= \Big|\mathrm{Re} \big( \bra{\eta}\bra{\Phi} U\otimes V \ket{\gamma}\ket{\Psi} \cdot \langle \emb_d \ket{\tilde{\emb}_d}\big) \Big|\\
&\geq \Big(1-\frac{2}{d}\Big)((1-\eps)\,\biasrk(\hat{G}))^{1/2}.
\end{align*}
Taking the limit as $\eps\to 0$ and $d\to\infty$ finishes the proof of the second part of the lemma.
\end{proof} 

We end this section with two examples illustrating both transformations described in the proofs of Lemmas~\ref{lem:rankone-equiv} and~\ref{lem:rankone-equiv-rev}.
We first illustrate the transformation $G\to \hat{G}$ from quantum XOR game to rank-one quantum game by applying it to the games $(T_n)$. 

\begin{topbotframe}
\begin{example}[From quantum XOR game to rank-one quantum game]\label{ex:xor-to-rk1}
Let $M$ be the matrix associated to the game $T_n$. Using the singular value decomposition $M=\frac{1}{2} \ket{00}\bra{\Psi_n^{me}} + 
\frac{1}{2} \ket{\Psi_n^{me}} \bra{00}$ in the proof of Claim~\ref{clm:gameformatrix}, we obtain the rank-one quantum game $\hat{T}_n$ defined by
$$ \ket{\eta} \,=\, \frac{1}{\sqrt{2}}\big( \ket{00}\ket{0} + \ket{\Psi_n^{me}}\ket{1}\big)\quad\text{and}\quad\ket{\gamma}\,=\,\frac{1}{\sqrt{2}}\big( \ket{\Psi_n^{me}}\ket{0} + \ket{00}\ket{1}\big).$$
The game $\hat{T}_2$ thus obtained is closely related to the ``coherent state exchange'' game introduced in~\cite{LTW08}, and we discuss this connection further in Section~\ref{sec:partial-transpose}. 
\end{example}
\end{topbotframe}

The next example shows how the other transformation, $\hat{G}\to G$, can be used to map a very simple rank-one quantum game $\hat{T}_n$ to the game $T_n$, leading us to rederive the fact that $\bias^*(T_n)=1$. Moreover, as we will show in Claim~\ref{claim:pt-infinite} an unbounded amount of entanglement between the players is \emph{necessary} in order for them to succeed with probability approaching $1$ in $T_n$. In contrast, the rank-one game $\hat{T_n}$ can be won perfectly using a maximally entangled state of dimension $n$.  Hence the example also demonstrates that our use of arbitrarily high-dimensional embezzlement states in the transformation $\hat{G}\to G$ given in the proof of Lemma~\ref{lem:rankone-equiv-rev} cannot be completely avoided. 

\begin{topbotframe}
\begin{example}[From rank-one quantum game to quantum XOR game]\label{ex:rk1-to-xor}
Consider the following simple rank-one quantum game $\hat{T}_n$ of size $n+1$. The referee first prepares the state $\ket{\eta} :=\ket{0}\ket{0}$, where here we think of each $\ket{0}$ as an $(n+1)$-dimensional state. He sends each player one of the two registers, and upon receiving their answers projects onto $P^{acc} = \ket{\gamma}\bra{\gamma}$, where $\ket{\gamma}=\ket{\Psi_n^{me}}$.
As a rank-one game, $\hat{G}$ is trivial, i.e., $\biasrk(\hat{G}) = 1$: an optimal strategy for the players, succeeding with probability $1$, consists in starting the game by sharing the state $\ket{\Psi} \,:=\, \ket{\gamma}$  and simply swapping their message register with their respective share of $\ket{\gamma}$.  

According to the transformation described in the proof of Lemma~\ref{lem:rankone-equiv-rev}, in the quantum XOR game $G$ that is obtained from $\hat{G}$ the players are asked to distinguish between the two states
$$ \ket{\psi_{\pm}} \,=\, \frac{1}{\sqrt{2}}\big(\ket{00} \ket{0}\ket{0} \pm \ket{11}\ket{\Psi_n^{me}}\big).$$
By a local unitary transformation, these two states are equivalent to the two states $\ket{\psi_0}$ and $\ket{\psi_1}$ used to define the game $T_n$. 
Hence Lemma~\ref{lem:rankone-equiv-rev} immediately reproves that $\bias^*(T_n)=1$ for all~$n$. 
\end{example}
\end{topbotframe}

%% file: sec_pt.tex
In this section we complete the proof of Theorem~\ref{thm:ptgamemain}. The first sequence of equalities were shown in Examples~\ref{ex:pt-unentangled-1} and~\ref{ex:pt-unentangled-2}, and the second in Example~\ref{ex:pt-bias}. It therefore remains to prove the ``moreover" part, namely, that perfect winning probability can only be achieved in the limit of infinite entanglement.  

\begin{claim}\label{claim:pt-infinite} Let $\eps>0$ be small enough, $n\ge 2$ and $d$ an integer. Suppose that $(A,B,\ket{\Psi})$, where $\ket{\Psi}\in\C^d\otimes \C^d$, is a strategy for the players in the game $T_n$ that achieves a bias at least $1-\eps$. Then $d \geq n^{C/\sqrt{\eps}}$, where $C>0$ is a universal constant.  
\end{claim}

We note that the bound in the claim is not far from tight, as Example~\ref{ex:pt-bias} demonstrates the existence of a strategy achieving an entangled bias of $1-\eps$ in $T_n$ and using an entangled state of dimension $d = n^{O(1/\eps)}$ for each player. 
We also note that one can derive this claim in a black-box fashion from the main result of~\cite{LTW08}. In more detail, the ``coherent state exchange'' game~\cite{LTW08} can be described as the rank-one game given by the states $\ket{\eta} = (\ket{00}\ket{0}+\ket{\varphi}\ket{1})/\sqrt{2}$ and $\ket{\gamma} = (\ket{00}\ket{0}+\ket{11}\ket{1})/\sqrt{2}$. This game is very close to the rank-one game $\hat{T}_2$ we associated to $T_2$ in Example~\ref{ex:xor-to-rk1}, and it is not hard to convert any strategy for $\hat{T}_2$ into a strategy for the coherent state exchange game with a similar success probability.%
\pnote{
Alice and Bob use the same entangled state $\ket{\Psi}$ as in their strategy for $T_2$. Let $\Pi$ be the projector on Span$\{\ket{1},\ket{2}\}$ and define
$$ A'\,:=\, \ket{1}\bra{0} A\Pi + \ket{0}\bra{0}\quad\text{and}\quad B'\,:=\, \ket{1}\bra{0} B\Pi + \ket{0}\bra{0}.$$
$A'$ and $B'$ have norm at most $1$, and using the equation in proof of Claim~\ref{claim:pt-infinite} we see that the success probability of the strategy $(A',B',\ket{\Psi})$ in $\hat{X}$ is at least\onote{strictly speaking in the rank-one section we never wrote the expression for the success probability of a strategy, so this is a bit ugly}
$$ \big| \bra{\Psi} \bra{\gamma} (A'\otimes B') \ket{\eta}\ket{\Psi} \big|^2 \,=\, \frac{1}{4}\big| 1 + \bra{\Psi}\bra{00} A\otimes B \ket{\varphi}\ket{\Psi} \big|^2  \,\geq\, (1-\eps)^2.$$
}
The claim for the case $n=2$ then follows from the main result of~\cite{LTW08}, and the general case can be derived from a straightforward modification of their proof. Below we give a more direct proof based on the techniques in~\cite{vDH03,LTW08}.

\begin{proof} The fact that $(A,B,\ket{\Psi})$ achieves a bias at least $1-\eps$ in $T_n$ implies, by definition,
\begin{equation*} %\label{eq:pt-infinite-1}
\big|\bra{\Psi} \big( (A\otimes B) (M(T_n)\otimes \Id) \big)\ket{\Psi}\big|\,=\,\big|\Re\big( \bra{00}\bra{\Psi} (A\otimes B) \ket{\me_n}\ket{\Psi}\big)\big| \,\geq\, (1-\eps).
\end{equation*} 
We use the the following fact, implicit in~\cite[Section~3]{LTW08}. (Its proof follows from the Fuchs-van de Graaf inequalities, which relate the fidelity to the trace norm, and Fannes' inequality, which provides a lower bound on the trace distance between two density matrices as a function of the difference of their von Neumann entropies.)

\begin{fact} Let $n,d$ be integers, $U,V\in\lin{\C^n\otimes \C^d}$ arbitrary operators of norm at most $1$, and $\ket{\varphi}\in\C^n\otimes \C^n$, $\ket{\Psi}\in\C^d\otimes \C^d$ of unit norm. Let $S$ be the von Neumann entropy of the reduced density of $\ket{\varphi}$ on any of the two subsystems, and assume $S\ge 1$. Then\pnote{How to get this bound: Fannes states that, if $x = \|\rho-\sigma\|_1$ and $S = |S(\rho)-S(\sigma)|$, then provided $x\le 1/e$ it holds that $S \le x\log (d/x)$. For us, $S\ge 1$, which implies that $x\ge 1/(3d)$ must hold; so we can simplify the bound to get $x\ge S/(2\log (3d))$ (this is quite crude in the constants, but gives the right order; it is also the same bound as~\cite{LTW08} use). We can then plug this in the equation from~\cite{LTW08} (middle of p.8).}
$$1-\big|\bra{\varphi}\bra{\Psi} U\otimes V \ket{0^n 0^n}\ket{\Psi}\big|^2 \,\geq\, \min\Big\{ \frac{1}{4e^2},\,\frac{S^2}{16 \log^2 (3d)}\Big\}.$$
\end{fact}

Applying the fact to $\ket{\varphi} = \ket{\me_n}$ and $U=A,V=B$, and using that the reduced density of the maximally entangled state $\ket{\me_n}$ has von Neumann entropy $\log n$, we obtain that the strategy $(A,B,\ket{\Psi})$ must satisfy $\eps \geq C^2\log^2(n)/\log^2(d)$ for some universal constant $C>0$. 
\end{proof}

%% file: sec_hi.tex
In this section we prove Theorem~\ref{thm:higamesummary}. The inequalities involving $\HI_1$ are proved in Sections~\ref{sec:hi-1} and~\ref{sec:hi-os}, and the inequalities involving $\HI_n$ are discussed in Section~\ref{sec:hi-general}. 

\subsubsection{The game \texorpdfstring{$\HI_1$}{H1}}\label{sec:hi-1}

The game $\HI:=\HI_1$ is a game of size $3$, 
whose associated game matrix $M=M(\HI)$ corresponds to the $n=1$ case of a family introduced in~\cite{HI95}. It is defined as
$$M \,=\, M(\HI) \,=\, \frac{1}{10} \big( C_1 \otimes C_1 + C_2 \otimes C_2 + C_3 \otimes C_3 \big),$$
where
$$C_1 = \begin{pmatrix} 0 & 0 & 0 \\ 0 & 0 & 1\\ 0 & -1 & 0 \end{pmatrix},\qquad C_2 = \begin{pmatrix} 0 & 0 & 1 \\ 0 & 0 & 0\\ -1 & 0 & 0 \end{pmatrix}\qquad \text{and}\qquad C_3=\begin{pmatrix} 0 & 1 & 0 \\ -1 & 0 & 0\\ 0 & 0 & 0 \end{pmatrix}.$$
One can verify that our choice of normalization is such that $\|M\|_1=1$. 
The following lemma sums up the results of~\cite{HI95} about $\HI$. 

\begin{lemma}[\cite{HI95}]\label{lem:hi-1}
The following hold for the game $\HI$:
$$\frac{2}{5}\,=\,\bias(\HI)\,=\,\bias^\C(\HI) \,<\, \biasme(\HI)  \,<\, \biasnc(\HI)\,=\,\frac{3}{5}.$$
\end{lemma}

For completeness, we prove this lemma in the three claims below, following the original arguments from~\cite{HI95}. 
We use the opportunity to observe that their proof of the equality $\biasnc(\HI)=3/5$ can be used to also show that $\biasos(\HI)=3/5$, as stated in Theorem~\ref{thm:higamesummary}. The only remaining inequality in Theorem~\ref{thm:higamesummary}, $\bias^*(\HI)<3/5$, is proved in Section~\ref{sec:hi-os}.

We start with the unentangled bias. 

\begin{claim}[\cite{HI95},~Remark~1.4] The unentangled and complex biases of the game $\HI$ are $\bias(\HI)=\bias^\C(\HI)=2/5$.
\end{claim}

\begin{proof} 
We first observe that it is easy to achieve a bias of $2/5$: for instance, Alice can use the observable $A = iC_1$, and Bob $B=-iC_1$, in which case $\Tr( M (A\otimes B) ) = (1/10) \Tr(C_1^2)^2 = 2/5$. Note that this strategy has the following operational interpretation: both players bet on their questions being $\{\ket{1}\pm i \ket{2}\}$, and measure their respective message registers in a basis containing both vectors. If they get neither they output a random answer. If Alice obtains $\ket{1}+ i \ket{2}$ she outputs $0$, and if she obtains $\ket{1}-i\ket{2}$ she outputs $1$; Bob does exactly the opposite. 
 
Now we show that $\bias^\C(\HI)\leq 2/5$. Consider an arbitrary complex strategy for the players, using (possibly non-Hermitian) matrices $A$ and $B$ with operator norm at most $1$. Using the Cauchy-Schwarz inequality, we may bound the bias that $(A,B)$ achieve in $\HI$ as follows:
\begin{align*}
\frac{1}{10} \sum_i \Tr\big( (A\otimes B) (C_i\otimes C_i)\big) 
 & = \frac{1}{10} \sum_i \Tr(A C_i)\Tr(B C_i) \nonumber \\
& \le \frac{1}{10} \Big(\sum_i \big|\Tr\big( AC_i \big)\big|^2\Big)^{1/2} \Big(\sum_i \big|\Tr\big( BC_i\big)\big|^2\Big)^{1/2}. %\label{eq:hi-un-1}
\end{align*}
We now show that $\sum_i \big|\Tr\big( AC_i \big)\big|^2 \le 4$ which, together with the analogous bound for $B$, would imply the claim. Notice that
\[
\sum_i \big|\Tr\big( AC_i \big)\big|^2 = \sum_i \Big|\Tr\Big( \frac{A-A^T}{2}\, C_i \Big)\Big|^2 = 4(|x|^2+|y|^2+|z|^2),
\]
where $x,y,z \in \C$ are defined by
\[
\frac{A-A^T}{2}\, =\, \begin{pmatrix} 0 & x & y \\ -x & 0 & z \\ -y & -z & 0 \end{pmatrix}.
\]
Observe that this matrix has rank at most 2 as its determinant is zero, and that its operator norm satisfies
$\|(A-A^T)/2\|_\infty\leq\|A\|_\infty\leq 1$. It therefore has at most two nonzero singular values, and both are at most $1$. Since the Frobenius norm is both the sum of squares of the singular values and the sum of the modulus squared of the entries, we conclude that $|x|^2+|y|^2+|z|^2 \le 1$.
\end{proof}

Next we show that the maximally entangled bias is strictly larger than the unentangled bias. 

\begin{claim}[\cite{HI95}, Theorem~3.4]\label{clm:himaxentlowerbound} The maximally entangled bias of the game $\HI$  satisfies $\biasme(\HI)\geq 5/9\approx 0.556$. 
\end{claim}

\begin{proof} We describe an explicit strategy. 
The players share the three-dimensional state $\ket{\Psi_3} = \frac{1}{\sqrt{3}} \big( \ket{11}+\ket{22}+\ket{33}\big)$. Upon receiving their question, each of them performs the same binary projective measurement. The projector corresponding to outcome $0$ projects on $\linspan \big\{ \ket{12}-\ket{21},\ket{13}-\ket{31},\ket{23}-\ket{31},\ket{\Psi_3}\big\}$. The projector corresponding to outcome $1$ projects on the $5$-dimensional orthogonal subspace.\footnote{We note that this strategy corresponds to measuring in the eigenbasis of the matrix $M$, and outputting the sign of the eigenvalue associated with the eigenvector obtained as outcome.}
One can directly compute the resulting value $5/9$, or a factor $\approx 1.389$ advantage over the best unentangled strategy. 
\end{proof}

We conclude by computing $\biasnc(\HI)$ and $\biasos(\HI)$. 

\begin{claim}\label{claim:hi-os} We have
$$\biasnc(\textsc{H})\,=\,\biasos(\textsc{H})\,=\,\frac{3}{5}. $$
\end{claim}

\begin{proof}
We first show that $\biasnc(\HI) \geq 3/5$. Let $(e_1,e_2,e_3)$ be the canonical basis of $\C^3$, and $\vec{X} = \vec{Y} = (1/\sqrt{2}) \sum_i e_i\otimes C_i \in \vecm{3}{\C^3}$. It is not hard to verify that $\vec{X}\vec{X}^\dagger = \vec{X}^\dagger \vec{X} = \Id$, so that $\vec{X},\vec{Y}$ satisfy the constraints~\eqref{eq:nc-gt-pcons}. Moreover, $\vec{X}\odot\vec{Y} = \sum_i C_i \otimes C_i /2$ so that the objective value in~\eqref{eq:nc-gt-p} is
$$\frac{1}{10} \Tr\big((\vec{X}\odot\vec{Y})\, M\big)\,=\, \frac{1}{20}\, \sum_{i,j} \,\Tr(C_i C_j)^2 \,=\, \frac{3}{5}.$$

Next we show $\biasos(\HI) \leq 3/5$. Let $(\vec{X}_R,\vec{X}_C,\vec{Y}_R,\vec{Y}_C)$ be vector-valued matrices satisfying the constraints~\eqref{eq:os-gt-pcons}. Let $X_j$ (resp.\ $Y_j$) be the matrix whose entries correspond to the $j$-th coordinate of the vector-entries of $\vec{X}_R$ (resp.\ $\vec{Y}_C$), so that $\vec{X}_R\odot\vec{Y}_C = \sum_j X_j \otimes Y_j$. By definition, the value achieved in~\eqref{eq:os-gt-p} is
\begin{align*}
\frac{1}{10}\sum_i \Tr\big( (\vec{X}_R \odot \vec{Y}_C) (C_i\otimes C_i)\big) &\leq \frac{1}{10} \Big(\sum_{i,j} \big|\Tr\big( X_jC_i \big)\big|^2\Big)^{1/2} \Big(\sum_{i,j} \big|\Tr\big( Y_jC_i\big)\big|^2\Big)^{1/2}\\
&\leq \frac{1}{5}\Tr\Big(\sum_j  X_jX_j^\dagger \Big)^{1/2} \Tr \Big( \sum_j Y_j^\dagger Y_j \Big)^{1/2}\\
&\leq\frac{3}{5}\Big\| \sum_j  X_jX_j^\dagger\Big\|_\infty^{1/2} \Big\| \sum_j  Y_j^\dagger Y_j \Big\|_\infty^{1/2}\\
& = \frac{3}{5}\Big\| \vec{X}_R \vec{X}_R^\dagger\Big\|_\infty^{1/2} \Big\| \vec{Y}_C^\dagger \vec{Y}_C \Big\|_\infty^{1/2}\,\leq\, \frac{3}{5},
\end{align*}
where the first inequality uses the Cauchy-Schwarz inequality, for the second we used the definition of $C_i$ and $|a-b|^2\leq 2(|a|^2+|b|^2)$ for any complex $a,b$, the third inequality uses $|\Tr(A)| \leq 3\|A\|_\infty$ for $3$-dimensional $A$, and the last follows from the constraints~\eqref{eq:os-gt-pcons}.
\end{proof}

\subsubsection{An upper bound on the entangled bias of \texorpdfstring{$\HI_1$}{HI}}\label{sec:hi-os}

In order to complete the proof of the part relating to $\HI_1$ in Theorem~\ref{thm:higamesummary}, it remains to show that $\bias^*(\HI) < 3/5$. 
This will be shown in Lemma~\ref{lem:hi-sdp} below.
Before getting there, we will show as a warm-up in Claim~\ref{claim:hi-sdp-exact} that no entangled strategy achieves a bias of exactly $3/5$. (Notice that this is a weaker statement since one might still have that for any $\eps>0$ there is a strategy with bias at least $3/5-\eps$.) Claim~\ref{claim:hi-sdp-exact} will not be directly used in the proof of Lemma~\ref{lem:hi-sdp}, but its proof is simpler and provides the template for the final argument. Before proceeding, we state a simple preliminary claim which lets us assume that the players' strategy has some symmetry that will be helpful in the analysis. 

\begin{claim}\label{claim:hi-symmetry}
Let $(A',B',\ket{\Psi'})$ be any strategy for the players in the game $\HI$. There exists a strategy $(A,A,\ket{\Psi})$ achieving a bias at least as high as that of $(A',B',\ket{\Psi'})$ in $\HI$, and such that $A$ is an observable and $\ket{\Psi}$ a permutation-invariant state whose reduced density $\rho = \Tr_{\HA} \ket{\Psi}\bra{\Psi} = \Tr_{\HB} \ket{\Psi}\bra{\Psi}$ on either player's private register has full support. 
\end{claim}

\begin{proof}
First, we restrict $A',B'$ and the state $\ket{\Psi'}$ to the support of 
$\Tr_{\HB} \ket{\Psi'}\bra{\Psi'}$ and $\Tr_{\HA} \ket{\Psi'}\bra{\Psi'}$. Denoting the resulting strategy by $(A'',B'',\ket{\Psi''})$,
we have that both reduced densities of $\ket{\Psi''}$ have full support, and the strategy achieves the same bias. 
Next, since the bias is a bilinear function of the two measurements, we can find a strategy $(A''',B''',\ket{\Psi''})$ achieving a bias at least as high and such that $A''',B'''$ are Hermitian with eigenvalues in $\{\pm 1\}$, i.e., they are observables. Finally, let 
$$\ket{\Psi}\,=\,\frac{1}{\sqrt{2}}\big(\ket{0}_A\ket{1}_B\otimes \ket{\Psi''}_{AB} + \ket{1}_A\ket{0}_B \otimes \ket{\Psi''_\tau}_{AB}\big),$$
where $\ket{\Psi''_\tau}$ is $\ket{\Psi''}$ with Alice and Bob's respective registers (which are of the same dimension) permuted, and define $A = \ket{0}\bra{0}\otimes A''' + \ket{1}\bra{1} \otimes B'''$. With these definitions, it is easy to verify that the strategy $(A,A,\ket{\Psi})$ has the same bias as $(A',B',\ket{\Psi'})$ in $\HI$, $A$ is an observable, and $\ket{\Psi}$ is a permutation-invariant state whose reduced density has full support.
\end{proof}

\begin{claim}\label{claim:hi-sdp-exact}
Let $(A,B,\ket{\Psi})$ be a strategy for the players in $\HI$. Then the corresponding bias is strictly less than $3/5$.  
\end{claim}

\begin{proof} 
Using Claim~\ref{claim:hi-symmetry} we may assume without loss of generality that $A=B$ and $\ket{\Psi}$ is a permutation-invariant state whose reduced density $\rho$ on either player has full support. For $i,j\in\{1,2,3\}$ let
\begin{equation}\label{eq:hi-aij-def}
 A_{ij} \,:=\, (\bra{i}\otimes \Id_{\HA})\, A\, (\ket{j}\otimes \Id_{\HA})\,\in\,\lin{\HA}.
\end{equation}
With this notation, and using that $A_{ij} = A_{ji}^\dagger$ since $A$ is Hermitian,
we can write the bias achieved by the strategy as
\begin{align}
\frac{1}{10} \sum_{i< j} \bra{\Psi} (A_{ij}-A_{ij}^\dagger)\otimes (A_{ij}-A_{ij}^\dagger) \ket{\Psi} 
&\stackrel{(a)}{\leq} \frac{1}{10} \sum_{i<j} \Tr\big((A_{ij}-A_{ij}^\dagger)(A_{ij}-A_{ij}^\dagger)^\dagger\rho\big)\notag \\
&\stackrel{(b)}{\leq} \frac{1}{10} \sum_{i<j} 2\,\Tr\big( (A_{ij}A_{ij}^\dagger + A_{ij}^\dagger A_{ij})\rho\big) \notag \\
& = \frac{1}{5} \,\sum_{i\neq j} \Tr\big( A_{ij}A_{ij}^\dagger\rho\big), \notag \\
&\stackrel{(c)}{\leq} \frac{3}{5}-\frac{1}{5} \,\sum_{i} \Tr\big( A_{ii}A_{ii}^\dagger\rho\big)
\stackrel{(d)}{\leq} \frac{3}{5}.\label{eq:hi-sdp-exact-2}
\end{align}
In (a) we apply the Cauchy-Schwarz inequality to the inner product between the vectors 
\begin{equation}\label{eq:hi-vectors}
 ((A_{ij}-A_{ij}^\dagger) \otimes \Id)^\dagger \ket{\Psi}\qquad\text{and}\qquad (\Id\otimes (A_{ij}-A_{ij}^\dagger)) \ket{\Psi},
\end{equation}
and we use that $\ket{\Psi}$ has the same reduced density $\rho$ on both subsystems. To obtain (b) observe that the difference between the two sides of the inequality is $\sum_{i<j}\Tr((A_{ij}+A_{ij}^\dagger)^2 \rho)/10$ which is clearly non-negative. The equality follows since $A$ is Hermitian. For (c), notice that the diagonal blocks of $AA^\dagger$ are given by $A_{ii}A_{ii}^\dagger + A_{ij}A_{ij}^\dagger + A_{ik}A_{ik}^\dagger$ for $i,j,k\in\{1,2,3\}$ all distinct, and since $AA^\dagger \leq \Id$, they must be at most $\Id_{\HA}$. 

To complete the proof, assume towards contradiction that all inequalities above are simultaneously tight. 
Then we have
\beq\label{eq:hi-sdp-exact-1}
\forall i\neq j\in\{1,2,3\},\qquad A_{ii}=0,\qquad A_{ij} = -A_{ij}^\dagger,\quad\text{and}\quad A_{ij}A_{ij}^\dagger = \Id/2,
\eeq
the first following from the tightness of (d), the second from the tightness of (b), and the third from the tightness of (c),
where in all three cases we also use the assumption that $\rho$ has full support. Now observe that there does not exist an $A$ that is both Hermitian and squares to identity, and satisfies these three conditions. Indeed, by considering the off-diagonal blocks in the equality $A^2 = \Id$ one obtains the equation $A_{ij}A_{ik}^\dagger = 0$ for $i,j,k\in\{1,2,3\}$ all distinct, which is easily seen to be incompatible with $A_{ij}A_{ij}^\dagger = \Id/2$ for every $i\neq j$ (since, say, the product of two nonsingular matrices is nonsingular). 
\end{proof}

\begin{lemma}\label{lem:hi-sdp}
There exists a $\delta>0$ such that $\bias^*(\HI) \leq 3/5-\delta$.
\end{lemma}

\pnote{One might hope to prove a $6-\delta$ upper bound on the right hand side of (a), i.e., without using the tightness of (a); this seems impossible since the right hand side of (a) \emph{can} be $6$ when $\rho$ is not full rank. Take e.g. 
$\rho = \ket{1}\bra{1}$ in a four dim space, $A_{12} = (0 1 0 0, -1 0 0 0, 0 0 0 0, 0 0 0 0)/ \sqrt{2}$, 
$A_{13} = (0 0 1 0, 0 0 0 0, -1 0 0 0, 0 0 0 0)/ \sqrt{2}$, $A_{23} = (0 0 0 1, 0 0 0 0, 0 0 0 0, -1 0 0 0)/ \sqrt{2}$}
We note that, while one could in principle extract a numerical estimate for $\delta$ from our proof, we have not attempted to do so, and in any case do not expect such an estimate to be tight. We also computed numerically the third level of a simple semidefinite hierarchy tightening the relaxation $\biasos$ along the lines of the general method given in~\cite{DLTW08}. This produced the upper bound $0.578$, but since we have not verified this bound carefully, it should be taken with a pinch of salt.  Finally, we remark that some basic numerical optimizations we performed failed to identify a strategy achieving bias higher than what is obtained in Claim~\ref{clm:himaxentlowerbound}, i.e., $5/9\approx 0.556$.

\begin{proof}
Let $(A,B,\ket{\Psi})$ be a strategy for the players in $\HI$. By Claim~\ref{claim:hi-symmetry} we may assume without loss of generality that $A=B$ and $\ket{\Psi}$ is a permutation-invariant state with $\rho$ its reduced density matrix.
The proof of the lemma is based on the following claim, which derives a quantitative version of the relations~\eqref{eq:hi-sdp-exact-1} that were used in the proof of Claim~\ref{claim:hi-sdp-exact}. 
It will be useful to introduce the notation $\nr{W} := \bra{\Psi}W W^\dagger\ket{\Psi}^{1/2}$ where $W \in M_d(\C) \otimes M_d(\C)$, $d$ being the dimension of the private space $\HA=\HB$ of each player. It is easy to see that $\nr{\cdot}$ is a semi-norm since it is derived from the semi-inner product $(W,Z) \mapsto \bra{\Psi}Z W^\dagger\ket{\Psi}$. For an $X \in M_d(\C)$ we often abuse notation and write $\nr{X}$ instead of $\nr{\Id \otimes X} = \nr{X \otimes \Id} = \Tr(XX^\dagger \rho)^{1/2}$.

\begin{claim}\label{claim:hi-sdp-1} Suppose that the strategy $(A,A,\ket{\Psi})$ achieves a bias at least $3/5-\delta$ in $\HI$. Then the following relations hold for all $i\neq j \neq k\in \{1,2,3\}$: 
\begin{align}
\nr{A_{ii}}^2 & \,=\, O(\delta),\label{eq:hi-sdp-1a}\\
\nr{A_{ij}+A_{ij}^\dagger}^2 &\,=\, O(\delta),\label{eq:hi-sdp-1c}\\
\sum_{i \neq j} \nr{A_{ij}}^2 &\ge 3-O(\delta), \label{eq:hi-sdp-goal} \\
\nr{\Id - A_{ii}A_{ii}^\dagger - A_{ij}A_{ij}^\dagger - A_{ik}A_{ik}^\dagger}^2&\,=\,O(\delta),\label{eq:hi-sdp-1b} \\
\nr{A_{ij} \otimes \Id - \Id \otimes A_{ij}^\dagger}^2 &\,=\,O(\delta),\label{eq:hi-sdp-commutate} 
\end{align}
where the $A_{ij}$ are as defined in~\eqref{eq:hi-aij-def}. 
\end{claim}

\begin{proof}
Since, by assumption, the first expression in~\eqref{eq:hi-sdp-exact-2} is at least $3/5-\delta$, all inequalities (a)--(d) should be tight up to an additive $\delta$. Eq.~\eqref{eq:hi-sdp-1a} follows immediately from the tightness of (d), \eqref{eq:hi-sdp-1c} follows from that of~(b), and~\eqref{eq:hi-sdp-goal} follows from that of the sequence~(c) and~(d). To show~\eqref{eq:hi-sdp-1b}, we use $0\leq \Id  - A_{ii}A_{ii}^\dagger - A_{ij}A_{ij}^\dagger- A_{ik}A_{ik}^\dagger \leq \Id$, where non-negativity follows by looking at the diagonal blocks in the inequality $AA^\dagger \leq \Id$, to bound
$$
 \nr{\Id  - A_{ii}A_{ii}^\dagger - A_{ij}A_{ij}^\dagger - A_{ik}A_{ik}^\dagger}^2 \,\leq\, \Tr\big( (\Id - A_{ii}A_{ii}^\dagger- A_{ij}A_{ij}^\dagger - A_{ik}A_{ik}^\dagger)\rho\big) \,=\, O(\delta),
$$
where the equality follows from the fact that the inequality (c) in~\eqref{eq:hi-sdp-exact-2} is tight up to $\delta$. 

To prove~\eqref{eq:hi-sdp-commutate}, observe that if $u$ and $v$ are two vectors of the same norm, and whose inner product is a non-negative real number that is at least $\|u\|^2-\eps$, then $\|u-v\| \le \sqrt{2\eps}$. Recall now that inequality (a) in~\eqref{eq:hi-sdp-exact-2} follows by applying, for each $i<j$, the Cauchy-Schwarz inequality to the two vectors in~\eqref{eq:hi-vectors}, which are of the same norm and whose inner product is a non-negative real number. It therefore follows from the tightness up to $\delta$ of this inequality that for all $i<j$,
$$ \nr{(A_{ij}-A_{ij}^\dagger) \otimes \Id + \Id\otimes (A_{ij}-A_{ij}^\dagger)} = O(\sqrt{\delta}).$$
Eq.~\eqref{eq:hi-sdp-commutate} now follows from two applications of the triangle inequality together with~\eqref{eq:hi-sdp-1c}.
The case $i>j$ follows since $A_{ij} = A_{ji}^\dagger$ and $\ket{\Psi}$ is permutation-invariant.
\end{proof}

Let $B_{ij}$ be the blocks of $AA^\dagger$; for all $i\neq j\neq k\in\{1,2,3\}$ we have
$$ B_{ii} = A_{ii}A_{ii}^\dagger + A_{ij}A_{ij}^\dagger + A_{ik}A_{ik}^\dagger\qquad\text{and}\qquad B_{ij}= A_{ii}A_{ji}^\dagger+A_{ij}A_{jj}^\dagger+A_{ik}A_{jk}^\dagger.$$
We will use the following estimates.

\begin{claim} Suppose that the strategy $(A,A,\ket{\Psi})$ achieves bias at least $3/5-\delta$ in $\HI$. Then the following relations hold:
\begin{align}
 \sum_i \nr{B_{ii} }^2&\geq 3-O\big(\sqrt{\delta}\big),\label{eq:hi-sdp-3a}\\
\sum_{i\neq j}  \nr{B_{ij} }^2&\geq \frac{3}{2} - O\big(\delta^{1/4}\big)\label{eq:hi-sdp-3b}.
\end{align}
\end{claim}

\begin{proof}
We will repeatedly use the following easy fact: if $X,Y,Y' \in M_d(\C) \otimes M_d(\C)$ are such $\|X\|_\infty$ is bounded by some universal constant, then 
\begin{align}
|\bra{\Psi} YX \ket{\Psi} - \bra{\Psi} Y'X \ket{\Psi}| &=
|\bra{\Psi} (Y-Y')X \ket{\Psi}| \notag\\
&\le \bra{\Psi} (Y-Y')(Y-Y')^\dagger \ket{\Psi}^{1/2}\cdot  \bra{\Psi} X^\dagger X \ket{\Psi}^{1/2} \notag\\
& \le O(\nr{Y-Y'}),\label{eq:hi-bounded}
\end{align}
where the first inequality is Cauchy-Schwarz. An analogous inequality holds with $Y$ and $Y'$ appearing to the right of $X$.

To prove~\eqref{eq:hi-sdp-3a}, use the triangle inequality and~\eqref{eq:hi-sdp-1b},
$$
\nrb{B_{ii} } =  \nrb{\Id - (\Id - B_{ii}) } \ge 1 - \nrb{\Id - B_{ii}} \ge 1-O(\sqrt{\delta}).
$$
To prove~\eqref{eq:hi-sdp-3b}, first fix some $i\neq j \in \{1,2,3\}$ and use the triangle inequality to obtain
\begin{align}
\nrb{B_{ij}} &\geq \nrb{A_{ik}A_{jk}^\dagger} - \nrb{A_{ii}A_{ji}^\dagger} - \nrb{A_{ij}A_{jj}^\dagger},\label{eq:hi-sdp-4}
\end{align}
where $k$ is the unique index in $\{1,2,3\}$ different from $i$ and $j$. Using $A_{ji}^\dagger A_{ji}\leq \Id$, the second term in~\eqref{eq:hi-sdp-4} can be bounded by $O(\sqrt{\delta})$ using~\eqref{eq:hi-sdp-1a},
$$\nrb{A_{ii}A_{ji}^\dagger} \,=\, \Tr\big( A_{ii}A_{ji}^\dagger A_{ji}A_{ii}^\dagger\rho\big)^{1/2} \,\leq\, \nr{A_{ii}} = O\big(\sqrt{\delta}\big).$$
The third term can be bounded as
\begin{align*}
\nrb{A_{ij}A_{jj}^\dagger}^2 &= \bra{\Psi} A_{ij}A_{jj}^\dagger A_{jj} A_{ij}^\dagger \otimes \Id \ket{\Psi} \\
&\approx \bra{\Psi} A_{jj}^\dagger A_{jj} A_{ij}^\dagger \otimes A_{ij}^\dagger \ket{\Psi} \approx 0 
\end{align*}
where both approximate equalities are up to an additive $O(\sqrt{\delta})$ and follow from~\eqref{eq:hi-bounded}:
in the first we replace $A_{ij} \otimes \Id$ with $\Id \otimes A_{ij}^\dagger$, using~\eqref{eq:hi-sdp-commutate} to bound the error,
and in the second we replace $A_{jj}^\dagger=A_{jj}$ with $0$, this time using~\eqref{eq:hi-sdp-1a} to bound the error.

To complete the proof, we will now show that
$$\sum_{i \neq j} \nrb{A_{ik}A_{jk}^\dagger}^2 \ge \frac{3}{2} - O(\delta^{1/2}), $$
where $k \in \{1,2,3\}$ is the unique index different from $i$ and $j$. 
Consider the two terms corresponding to $(i,j)=(1,2)$ and $(i,j)=(1,3)$. For the former, we write
\begin{align*}
\nrb{A_{13}A_{23}^\dagger}^2 &= \bra{\Psi} A_{13}A_{23}^\dagger A_{23} A_{13}^\dagger \otimes \Id \ket{\Psi} \\
&\approx \bra{\Psi} A_{23}^\dagger A_{23} A_{13}^\dagger \otimes A_{13}^\dagger \ket{\Psi} \\
&\approx \bra{\Psi} A_{23}^\dagger A_{23} \otimes A_{13}^\dagger A_{13} \ket{\Psi},
\end{align*}
where as before the approximate equalities are up to an additive $O(\sqrt{\delta})$ and follow
from~\eqref{eq:hi-bounded} and~\eqref{eq:hi-sdp-commutate}.
For the latter we follow the same sequence, but add two extra steps at the end, 
\begin{align*}
\nrb{A_{12}A_{32}^\dagger}^2 
&\approx \phantom{-}\bra{\Psi} A_{32}^\dagger A_{32} \otimes A_{12}^\dagger A_{12} \ket{\Psi} \\
&\approx - \bra{\Psi} A_{23}^\dagger A_{32} \otimes A_{12}^\dagger A_{12}\ket{\Psi} \\
&\approx \phantom{-}\bra{\Psi} A_{23}^\dagger A_{23} \otimes A_{12}^\dagger A_{12}\ket{\Psi},
\end{align*}
where the last two approximate equalities are up to an additive $O(\sqrt{\delta})$ and follow
from~\eqref{eq:hi-bounded} and~\eqref{eq:hi-sdp-1c}.
Summing the two terms, and using~\eqref{eq:hi-sdp-1b} and~\eqref{eq:hi-sdp-1a}, we get
\begin{align*}
\nrb{A_{13}A_{23}^\dagger}^2 + \nr{A_{12}A_{32}^\dagger}^2 
&\approx \bra{\Psi}A_{23}^\dagger A_{23} \otimes (A_{13}^\dagger A_{13}+A_{12}^\dagger A_{12})  \ket{\Psi}\\
&\approx \bra{\Psi} A_{23}^\dagger A_{23} \otimes \Id \ket{\Psi} \\
&= \nr{A_{32}}^2.
\end{align*}
Repeating the same proof for the other two pairs of terms, summing the results, and noticing that $\nr{A_{ij}} \approx \nr{A_{ji}}$ due
to~\eqref{eq:hi-sdp-1c}, we get that
\[
\sum_{i \neq j} \nrb{A_{ik}A_{jk}^\dagger}^2 \ge \frac{1}{2} \sum_{i \neq j}\nr{A_{ij}}^2 - O(\sqrt{\delta}) \ge \frac{3}{2} - O(\sqrt{\delta}), 
\]
where the last inequality uses~\eqref{eq:hi-sdp-goal}. This complete the proof. 
\end{proof}

To conclude the proof of the lemma, we claim that~\eqref{eq:hi-sdp-3a} and~\eqref{eq:hi-sdp-3b} are incompatible with the condition $(AA^\dagger)^2 \leq \Id$. To see why, let $D_{ii}$ be the diagonal blocks of $(AA^\dagger)^2$ and note that together both inequalities imply that
$$ \sum_i \Tr\big(D_{ii} \rho\big) \,\geq\, 3 + \frac{3}{2} - O\big(\delta^{1/4}\big).$$
But $(AA^\dagger)^2 \leq \Id$ implies that 
$$ \sum_i \Tr\big(D_{ii}\rho\big) \,\leq\, 3,$$
which gives a contradiction for small enough $\delta$. 
\end{proof}

\subsubsection{The games \texorpdfstring{$(\HI_n)$}{Hn}}\label{sec:hi-general}

In this section we follow~\cite{HI95} in introducing a family of games $\{\HI_n\}_{n\geq 1}$ which generalizes the game $\HI=\HI_1$ from the previous section. The main motivation for studying this family is that it satisfies $\lim_{n\to\infty} \biasnc(\HI_n)/\bias^\C(\HI_n) = 2$, which as shown in Theorem~\ref{thm:unentangled-bias} is as strong a gap as possible between these two quantities.\footnote{In this case it also holds that $\bias(\HI_n)=\bias^\C(\HI_n)$, and we do not know if there exist games for which $\biasnc/\bias>2$.} 

For any integer $n\geq 1$, $\HI_n$ is a quantum XOR game of size $N=\binom{2n+1}{n}$. To describe $\HI_n$,\footnote{See~\cite[Section 4]{HI95} for an alternative definition.} it will be convenient to index the canonical basis of $\C^N$ by subsets $S\subseteq [2n+1]$ of cardinality $n$. Let $i\in[2n+1]$, and for every $S\subseteq [2n+1]$, $|S|=n$, such that $i\notin S$ let $\eps(i,S)$ be the sign of the permutation of $[2n+1]$ defined (using the standard one-line notation) as $(S,\, i,\,\overline{S\cup \{i\}})$, where both $S$ and $\overline{S\cup i}$ are ordered. Define a linear map $c_i$ from $\C^N$ to itself by
$$ c_i :\, e_S \,\mapsto\, \begin{cases} \eps(i,S)\, e_{\overline{S\cup\{i\}}} &\text{ if $i\notin S$,}\\ 0 & \text{ otherwise.}\end{cases}$$
Let $C_i$ be the matrix of $c_i$ in the basis $\{e_S\}$, and define
\begin{equation*}%\label{eq:def-him}
M_n \,:=\,\binom{4n+1}{2n}^{-1}\, \sum_{i=1}^{2n+1} C_i \otimes C_i.
\end{equation*}
It is not hard to check that $M_n$ is Hermitian (in fact, $C_i$ is real symmetric or anti-symmetric, depending on the parity of $n$). Moreover, the normalization factor is chosen so as to ensure $\|M_n\|_1=1$~\cite[Lemma~3.2]{HI95}.
Let $\HI_n$ be the quantum XOR game whose associated matrix is $M(\HI_n)=M_n$. One can check that the game $\HI_1$ corresponds to the game $\HI$ from the previous section. Haagerup and Itoh showed the following (see also~\cite[Section~11]{PisierGT} for the first part of the lemma). 

\begin{lemma}[\cite{HI95}]\label{lem:hi-general} For every integer $n\geq 1$ it holds that
$$ \bias(\HI_n) \,=\,\bias^\C(\HI_n)\,=\, \Big(\frac{n+1}{2n+1}\Big)^2 \binom{2n+1}{n}^2\binom{4n+1}{2n}^{-1}\quad\text{and}\quad \frac{\biasnc(\HI_n)}{\bias^\C(\HI_n)} \,=\, \frac{2n+1}{n+1}.$$
Moreover, 
$$ \frac{\biasme(\HI_n)}{\bias^\C(\HI_n)} \,\geq\, \frac{(2n+1)^2}{(n+1)^3} \binom{4n+1}{2n}^2\binom{2n+1}{n}^{-4}\,\to_{n\to \infty}\, \frac{\pi}{2}.$$
\end{lemma}

Actually,~\cite{HI95} do not consider the (real) bias $\bias(\HI_n)$; the fact that $\bias(\HI_n) =\bias^\C(\HI_n)$ mentioned above is an easy observation that follows from their lower bound on $\bias^\C(\HI_n)$. The relevance of the lower bound on $\biasme(\HI_n)/\bias^\C(\HI_n)$ is that it is strictly larger than the complex Grothendieck constant $K_G^\C \leq 1.405$, which bounds the same ratio from above for all \emph{diagonal} matrices $M$, as described in Theorem~\ref{thm:class-bias}.

%% file: sec_grothendieck.tex
In this section we introduce Grothendieck's original ``commutative'' inequality and two of its more recent generalizations, explaining how they lead to the key inequalities in Theorem~\ref{thm:class-bias}, Theorem~\ref{thm:unentangled-bias} and Theorem~\ref{thm:entangled-bias} respectively. 
\subsection{The commutative Grothendieck inequality}\label{sec:com-gt}

Grothendieck proved his famous inequality in~\cite{Gro53} motivated by the study of norms on the tensor product of two Banach spaces. Many equivalent formulations of the inequality exist, and we refer to Sections~2 and~3 of~\cite{PisierGT} for a comprehensive survey. 
Here we restrict our attention to the finite-dimensional case. 

In its most concrete form,\footnote{Grothendieck's inequality was first reformulated in this way by Lindenstrauss and Pe{\l}czy\'{n}ski~\cite{Linden68}.} the inequality says that there exists a universal constant $K_G^\R$ such that, for any integer $n$ and any real $R = (R_{s,t})_{s,t\in[n]}$,
\begin{equation}\label{eq:commgroth}
\sup_{\substack{d,\,x_s,y_t\in \R^d}}\, \Big|\sum_{s,t} R_{st} \,\langle x_s,y_t \rangle\Big|\,\leq\, K_G^\R \,\max_{x_s,y_t\in [-1,1]}  \Big|\sum_{s,t} R_{st}\, x_sy_t\Big|,
\end{equation}
where the supremum on the left-hand side is taken over all vectors satisfying
$$\max\Big\{ \max_s \|x_s\|^2,\, \max_t \|y_t\|^2 \Big\}\, \leq\, 1.$$
Recalling~\eqref{eq:class-bias} and the discussion following it, we see that the supremum in the right-hand side of~\eqref{eq:commgroth} is  $\bias(R)$.
Also, the left-hand side is $\biascom(R)$ as defined in~\eqref{eq:class-bias-sdp}. The fact that the supremum in~\eqref{eq:commgroth} is taken over vectors in $\R^d$ and not $\C^d$ does not change the supremum as already mentioned after~\eqref{eq:class-bias-sdp}.
Therefore, the inequality is equivalent to the bound $\biascom(R)\leq K_G^\R\bias(R)$, for any real $R$, claimed in Theorem~\ref{thm:class-bias}.

We now give an equivalent formulation of Grothendieck's inequality. In order to emphasize the connection to the non-commutative generalizations of Grothendieck's inequality, we follow the notation introduced in Section~\ref{sec:quantum-xor}, even if it is somewhat artificial in the present context. 
Recall the notation introduced above Definition~\ref{def:bias-nc}.
We associate with a classical XOR game with coefficients $R = (R_{s,t})_{s,t\in[n]}$ a diagonal matrix $M$ in $M_n(\C) \otimes M_n(\C)$. With these conventions, Grothendieck's inequality~\eqref{eq:commgroth} is easily seen to be equivalent to the following. 

\begin{theorem}[\cite{Gro53}]\label{thm:gt-commutative} There exists a universal constant $K_G^\R$ such that, for any integer $n$ and real diagonal $n^2\times n^2$ matrix $M$, 
\begin{align*}%\label{eq:c-gt-p}
\sup_{d,\,\vec{X},\vec{Y}\in \matv{d}{\C^n}} \big| \Tr\big( (\vec{X}\odot \vec{Y})\,M\big) \big| \,\leq\, 
   K_G^\R \sup_{\substack{X,Y\in \herm{\C^n}, \\ \|X\|_\infty,\|Y\|_\infty\leq 1}} \big| \Tr\big( ( X\otimes Y)\,M\big)\big|,
\end{align*}
where the supremum on the left-hand side is taken over all $d \ge 1$ and vector-valued matrices $\vec{X},\vec{Y}\in \matv{d}{\C^n}$ such that
\begin{align*}%\label{eq:c-gt-p-cons}
\max\Big\{ \big\| \vec{X}\vec{X}^\dagger \big\|_{\infty},\, \big\|\vec{Y}\vec{Y}^\dagger \big\|_{\infty} \Big\} \,\leq\, 1.
\end{align*}
\end{theorem}

To see that the theorem is equivalent to the formulation~\eqref{eq:commgroth}, the main thing to observe is that since $M$ is diagonal, the suprema above can be equivalently restricted to diagonal $\vec{X}$, $\vec{Y}$ (resp.\ $X$, $Y$). As argued before the theorem, the fact that $\vec{X},\vec{Y}$ may have complex vectors on their diagonal does not change the supremum on the left-hand side. Moreover, the constraint that $X,Y$ are Hermitian of norm at most $1$ implies that their diagonal entries are real numbers in $[-1,1]$. 

The above discussion relates to what is known as the \emph{real} Grothendieck inequality. 
If we relax the supremum on the right-hand side to allow all $X,Y\in M_n(\C)$ with norm at most $1$, and extend the inequality to all complex diagonal $M$, we obtain what is known as the \emph{complex} Grothendieck inequality, which holds with the better constant $K_G^\C < K_G^\R$. This inequality corresponds to the second statement in Theorem~\ref{thm:class-bias}. The non-commutative Grothendieck inequalities described below are in this complex setting. 

\subsection{The non-commutative Grothendieck inequality}\label{sec:nc-gt}

The non-commutative Grothendieck inequality generalizes Theorem~\ref{thm:gt-commutative} to the non-commutative setting by replacing the space $\ell_\infty$ by an arbitrary $C^*$-algebra. Originally conjectured by Grothendieck, it was first proved by Pisier~\cite{Pisier78NCGT} for $C^*$-algebras satisfying a certain ``approximability'' assumption, and in the general case by Haagerup~\cite{Haagerup85NCGT}. Here we restrict our attention to the finite-dimensional
case of $M_n(\C)$.

\begin{theorem}[\cite{Pisier78NCGT,Haagerup85NCGT}]\label{thm:nc-gt} Let $n$ be an integer and $M\in M_n(\C)\otimes M_n(\C)$. Then
\begin{align}\label{eq:nc-gt-p-2}
\sup_{d,\,\vec{X},\vec{Y}\in\matv{d}{\C^n} }  \big|\Tr\big( (\vec{X}\odot \vec{Y})\, M\big) \big|  \,\leq\, 2 \,\sup_{\substack{X, Y\in M_n(\C),\\ \|X\|_{\infty}\leq 1,\,\|Y\|_{\infty}\leq 1}} \big| \Tr\big( (X\otimes Y)\, M\big)\big|,
\end{align}
where the supremum on the left-hand side is taken over all $d \ge 1$ and vector-valued matrices $\vec{X},\vec{Y}\in\matv{d}{\C^n}$ such that
\begin{align}\label{eq:nc-gt-p-cons-2}
\max\Big\{\, \big\|\vec{X} \vec{X}^\dagger \big\|_{\infty}+ \big\|  \vec{X}^\dagger \vec{X} \big\|_{\infty},\,\big\|\vec{Y} \vec{Y}^\dagger \big\|_{\infty} + \big\|\vec{Y}^\dagger \vec{Y} \big\|_{\infty} \,\Big\}\,\leq\, 2.
\end{align}
\end{theorem}

Let $M$ be the Hermitian matrix associated with a given quantum XOR game $G$. Since the constraint~\eqref{eq:nc-gt-p-cons-2} is less restrictive than~\eqref{eq:nc-gt-pcons}, the left-hand side of~\eqref{eq:nc-gt-p-2} is at least as large as the bias $\biasnc(M)$. Moreover, the right-hand side is exactly twice the complex bias $\bias^\C(M)$. Hence the inequality $\biasnc(G)\leq 2\bias^\C(G)$ claimed in Theorem~\ref{thm:unentangled-bias} is a direct corollary of Theorem~\ref{thm:nc-gt}.

We formulated Theorem~\ref{thm:nc-gt} in a slightly different way than it appears in the literature on the subject. For the benefit of the interested reader, we briefly explain why our statement is equivalent to (the finite dimensional version of) the one that appears in Theorem~7.1, Eq.~(7.2) of the survey~\cite{PisierGT}. 
The first thing to observe is that relaxing the constraint~\eqref{eq:nc-gt-p-cons-2} by taking the geometric average of the two terms instead of their maximum
results in an equivalent statement. The reason is that one can replace $\vec{X}$ and $\vec{Y}$ by $s\vec{X}$ and $s^{-1}\vec{Y}$ for some $s>0$ so that the two terms are equal and this does not affect the left-hand side of~\eqref{eq:nc-gt-p-2}. Once this modification is done, the only remaining thing to observe is that the matrix $M \in M_n(\C)\otimes M_n(\C)$ can be equivalently thought of as the bilinear form $\varphi: M_n(\C) \times M_n(\C) \to \C$
defined by $\varphi(E_{i,j},E_{k,l}) = M_{(j,l),(i,k)}$ where $E_{i,j}=\ket{i}\bra{j}$ is the canonical basis of $M_n(\C)$. With this identification, we have that for all $X,Y$, $\varphi(X,Y)=\Tr\big( (X\otimes Y)\, M\big)$. Hence, by definition, the right-hand side of~\eqref{eq:nc-gt-p-2} is $\|\varphi\|$, the norm of the bilinear form $\varphi$, and now the theorem is easily seen to be equivalent to the one in~\cite{PisierGT}.

\subsection{The operator space Grothendieck inequality}\label{sec:os-gt}

More recent work has lead to a different generalization of Grothendieck's inequality. This generalization originates in the study of operator spaces, and focuses on the so-called \emph{(jointly) completely bounded} norm of a bilinear form, as opposed to the \emph{norm} that appears in the non-commutative Grothendieck inequality above. Two variants of this ``operator space Grothendieck inequality" are known, one by Pisier and Shlyakhtenko~\cite{PS02OSGT} (which applies to ``exact" operator spaces) and another by Haagerup and Musat~\cite{HM08} (for not necessarily exact $C^*$-algebras). Here we state the special case of $M_n(\C)$, which is all that is needed for our purposes.

\begin{theorem}[\cite{PS02OSGT,HM08}]\label{thm:os-gt} Let $n$ be an integer, and $M\in M_n(\C)\otimes M_n(\C)$. Then 
\begin{align}\label{eq:os-gt-p-2}
\sup_{d,\,\vec{X}_R,\vec{X}_C,\vec{Y}_R,\vec{Y}_C\in\matv{d}{\C^n}}  \big|\Tr\big( (\vec{X}_R\odot \vec{Y}_C)\, M\big) \big|  \,\leq\, 2 \, \sup_{\substack{d,\,X,Y\in M_{nd}(\C),\\ \|X\|_\infty,\|Y\|_\infty \leq 1}} \big\|\Tr_{\C^n\otimes \C^n}\big(  (X\otimes Y)\,(M\otimes \Id_{\C^d\otimes \C^d}) \big)\big\|_\infty,
\end{align}
where the supremum on the left-hand side is taken over all $d\ge 1$ and vector-valued matrices $\vec{X}_R,\vec{X}_C,\vec{Y}_R,\vec{Y}_C\in \matv{d}{\C^n}$ such that $\vec{X}_R\odot \vec{Y}_C = \vec{X}_C\odot \vec{Y}_R$ and 
\begin{equation}\label{eq:os-gt-p-cons-2}
 \big\|\vec{X}_R \vec{X}_R^\dagger \big\|_{\infty}^{1/2}\, \big\|  \vec{Y}_C^\dagger \vec{Y}_C \big\|_{\infty}^{1/2}+\big\|\vec{X}_C^\dagger \vec{X}_C \big\|_{\infty}^{1/2} \,\big\|\vec{Y}_R \vec{Y}_R^\dagger \big\|_{\infty}^{1/2} \,\leq\, 2.\footnote{An elementary manipulation shows that replacing this constraint by $\max\big\{ \|\vec{X}_R \vec{X}_R^\dagger \|_{\infty} + \|\vec{X}_C^\dagger \vec{X}_C \|_{\infty}, \|\vec{Y}_R \vec{Y}_R^\dagger \|_{\infty} + \|\vec{Y}_C^\dagger \vec{Y}_C \|_{\infty}\big\}\leq 2$ does not change the value of the supremum on the left-hand side of~\eqref{eq:os-gt-p-2}.}
\end{equation}
\end{theorem}

One can easily verify that the left-hand side of~\eqref{eq:os-gt-p-2} is always at least as large as $\biasos(G)$, as the constraint~\eqref{eq:os-gt-p-cons-2} is less restrictive than~\eqref{eq:os-gt-pcons}.\footnote{We note that the quantity $\biasos$ coincides with what is known in operator space theory as the symmetrized Haagerup norm on $M_n \otimes M_n$; see~\cite[Chapter 5]{pisierbook}.\pnote{We have the sequence of inequalities $cb \le \gamma_{R\oplus C} \le \omega^{os}=\mu = \ell_\infty sdp \le usual \ell_2 sdp \le 2 cb$. See Pisier's book 5.23 for the second (easy) inequality.}\pnote{Can our tighter NCGT relaxation be seen as some kind of factorization norm?}} Moreover, if $M$ is the Hermitian matrix associated with a given quantum XOR game $G$, then the supremum on the right-hand side is exactly the entangled bias $\bias^*(G)$. Therefore, the inequality $\biasos(G)\leq 2\bias^*(G)$ claimed in Theorem~\ref{thm:entangled-bias} is a direct corollary of Theorem~\ref{thm:os-gt}.

\pnote{Our way of doing things using GSVD seems to replace Pisier's ``trick of independent interest" in the proof of his Prop 18.2}
For the benefit of the interested reader, we conclude this section by briefly explaining how Theorem~\ref{thm:os-gt} can be deduced from results appearing in the literature. In the first and main step we modify the supremum in the left-hand side of~\eqref{eq:os-gt-p-2} which, to recall, is taken over the set of elements in 
\[
{\mathcal W} = \Big\{ (\vec{X}_R,\vec{X}_C,\vec{Y}_R,\vec{Y}_C) \in (\matv{d}{\C^n})^4 ~|~ d \ge 1, ~~ \vec{X}_R\odot \vec{Y}_C = \vec{X}_C\odot \vec{Y}_R \Big\}
\]
satisfying~\eqref{eq:os-gt-p-cons-2}. Consider now the set
%\onote{old definition, seems unnecessary:
%\begin{align*}
%{\mathcal W'} = \Big\{ (\vec{X}_R,\vec{X}_C,\vec{Y}_R,\vec{Y}_C) &\in (\matv{d}{\C^n})^4 ~|~ d \ge 1, ~~ \vec{X}_R = \vec{s} \vec{X}, \vec{X}_C = \vec{t}\vec{X}, \vec{Y}_R = \vec{s}\vec{Y}, \vec{Y}_C = \vec{t} \vec{Y} \\
%&\mbox{~for some~} \vec{X},\vec{Y} \in \matv{d}{\C^n} \mbox{~~and~~} \vec{s},\vec{t} \in (\R^{\ge 0})^d \Big\},
%\end{align*}
%where $\vec{t}\vec{X}$ denotes the sequence $(t_iX_i)$.
%Equivalently, ${\mathcal W'}$ consists of all four-tuples in which for each $i \in \{1,\ldots,d\}$, $((\vec{X}_R)_i,(\vec{Y}_R)_i)$ is ``non-negatively proportional" to $((\vec{X}_C)_i,(\vec{Y}_C)_i)$, i.e., one is a product of the other by a non-negative number. }
\begin{align*}
{\mathcal W'} = \Big\{& (\vec{X}_R,\vec{X}_C,\vec{Y}_R,\vec{Y}_C) \in (\matv{d}{\C^n})^4 ~|~ d \ge 1, \\
& \quad \forall i \in \{1,\ldots,d\}, ((\vec{X}_R)_i,(\vec{Y}_R)_i) \text{~is ``non-negatively proportional" to~} ((\vec{X}_C)_i,(\vec{Y}_C)_i) \Big\},
\end{align*}
where non-negatively proportional means that one is obtained as the product of the other by a non-negative number (or equivalently, that either the latter is a product of the former by a positive number, or at least one of the two is zero).
It is easy to check that ${\mathcal W'} \subseteq {\mathcal W}$ and hence if we modify the supremum to be over the elements of ${\mathcal W'}$ (satisfying~\eqref{eq:os-gt-p-cons-2}), it is not greater than the original one. We claim that the modified supremum is in fact equal to the original one. To show this, we first describe the so-called generalized singular value decomposition (see also~\cite[Theorem 4.2.2]{Bjorck96book} for the proof) and then derive from it a claim.

\begin{fact}[\cite{vanLoan75,Paige80}]\label{fact:gsvd} Let $A_1$ and $A_2$ be two $n\times d$ matrices for some $n\leq d$. Let $k$ be the rank of $\begin{pmatrix} A_1 & A_2 \end{pmatrix}$. Then there exist $d\times d$ unitaries $U_1,U_2$, an $n\times k$ matrix $R$ of full column rank, and non-negative diagonal matrices $D_1, D_2$ of dimension $k\times k$ such that $D_1^2 + D_2^2 = \Id$ satisfying
\begin{align}\label{eq:gsvd}
 A_1 U_1 = R \begin{pmatrix} D_1 & 0_{k\times (d-k)} \end{pmatrix} \qquad\text{and}\qquad A_2 U_2 = R \begin{pmatrix} D_2 & 0_{k\times (d-k)} \end{pmatrix}.
\end{align}
\end{fact}

\begin{claim}\label{clm:gsvd}
Let $A_1$, $A_2$, $B_1$, and $B_2$ be $n\times d$ matrices satisfying $A_1 B_1^\dagger = A_2 B_2^\dagger$. Then for some $d' \ge d$ there are $d \times d'$ isometries $V_1$, $V_2$ (i.e., matrices with orthonormal rows) such that if we denote by $(A)_i$ the $i$th column of a matrix $A$, then for all $i \in \{1,\ldots,d'\}$, $((A_1 V_1)_i,(B_2 V_2)_i)$ is non-negatively proportional to $((A_2 V_2)_i,(B_1 V_1)_i)$.
\end{claim}
\begin{proof}
First, we can assume that $d \ge n$, as otherwise we can use the isometry to append zero coordinates. Now apply Fact~\ref{fact:gsvd} to $A_1$ and $A_2$, and let $U_1$, $U_2$, $R$, $D_1$, $D_2$ be the resulting matrices. Our assumption then implies that
$$ R \begin{pmatrix} D_1 & 0_{k\times (d-k)} \end{pmatrix} (B_1 U_1)^\dagger = R \begin{pmatrix} D_2 & 0_{k\times (d-k)} \end{pmatrix} (B_2 U_2)^\dagger,$$
which, since $R$ is of full column rank, implies 
$$ (B_1 U_1) \begin{pmatrix} D_1 \\ 0_{(d-k) \times k} \end{pmatrix} = (B_2 U_2) \begin{pmatrix} D_2 \\ 0_{(d-k) \times  k} \end{pmatrix} .$$
Together with~\eqref{eq:gsvd}, we obtain
$$ \begin{pmatrix} A_2 U_2 \\ B_1 U_1 \end{pmatrix} \begin{pmatrix} D_1 \\ 0_{(d-k) \times k} \end{pmatrix} = \begin{pmatrix} A_1 U_1 \\ B_2 U_2 \end{pmatrix}  \begin{pmatrix} D_2 \\ 0_{(d-k) \times  k} \end{pmatrix},$$
and therefore, since $D_1$ and $D_2$ are never both zero at the same location, $((A_1 U_1)_i,(B_2 U_2)_i)$ is non-negatively proportional to $((A_2 U_2)_i,(B_1 U_1)_i)$ for all $i \in \{1,\ldots,k\}$. For $i>k$, observe that both $(A_1 U_1)_i$ and $(A_2 U_2)_i$ are zero. Therefore, we can complete the proof by taking $V_1 = U_1 U'_1$, $V_2 = U_2 U'_2$, for the isometries $U'_1, U'_2$ defined as 
$$ U'_1 = \begin{pmatrix} \Id & 0 & 0 \\ 0 & \Id & 0 \end{pmatrix}, 
   \qquad U'_2 = \begin{pmatrix} \Id & 0 & 0 \\ 0 & 0 & \Id \end{pmatrix}  ,$$
where the dimensions of the three column blocks are $k,d-k,d-k$ respectively and those of the row blocks are $k,d-k$. 
\end{proof}

To show that the modified supremum is equal to the original one as claimed, take any tuple $(\vec{X}_R,\vec{X}_C,\vec{Y}_R,\vec{Y}_C) \in {\mathcal W}$ and apply Claim~\ref{clm:gsvd} with $A_1$, $A_2$, $B_1$, and $B_2$ taken to be the $n^2 \times d$ matrices whose rows contain the vector entries of $\vec{X}_R$, $\vec{X}_C$, $\vec{Y}_C$, and $\vec{Y}_R$ respectively. The condition $A_1 B_1^\dagger = A_2 B_2^\dagger$ holds because it is equivalent to $\vec{X}_R\odot \vec{Y}_C = \vec{X}_C\odot \vec{Y}_R$. The claim then shows that there exist two isometries, such that if we apply one to the vector entries of $\vec{X}_R$ and of $\vec{Y}_C$, and the other to the vector entries of $\vec{X}_C$ and of $\vec{Y}_R$, then the resulting tuple $(\vec{X}'_R,\vec{X}'_C,\vec{Y}'_R,\vec{Y}'_C)$ is in ${\mathcal W'}$. It remains to notice that since all we did was apply isometries, the new tuple achieves the same goal function and still satisfies the constraint~\eqref{eq:os-gt-p-cons-2}. 

In the second step, consider the set ${\mathcal W''}$ defined like ${\mathcal W'}$ except that we require \emph{positive} proportionality instead of a non-negative one (which means that one pair can be obtained from the other by multiplication by a positive number). By continuity of the 
goal function and the constraint~\eqref{eq:os-gt-p-cons-2}, it is clear that taking the supremum over ${\mathcal W''}$ is again equivalent to the original form. 

The resulting equivalent form of Theorem~\ref{thm:os-gt} (with the supremum over ${\mathcal W''}$ instead of over ${\mathcal W}$) is essentially the way the theorem appears in the literature, although using different terminology, as we now explain in more detail. As we saw in the previous section, $M$ can be equivalently thought of as a bilinear form $\varphi: M_n(\C) \times M_n(\C) \to \C$. One can also think equivalently of $M$ as a linear map $u$ defined by $\langle u x, y \rangle = \varphi(x,y)$ for all $x,y \in M_n(\C)$, and this is the terminology usually adopted in~\cite{PisierGT}. The supremum on the right-hand side of~\eqref{eq:os-gt-p-2} is known in operator space theory as the \emph{completely bounded} norm of $u$, or equivalently, the \emph{jointly completely bounded} norm of $\varphi$.\footnote{In case one replaces $\Id_{\C^d\otimes \C^d}$ in that expression by the rank-one projector $\ket{\Psi_d^{me}}\bra{\Psi_d^{me}}$ on the maximally entangled state, one obtains the notion of \emph{tracial boundedness} introduced in~\cite{Itoh87} (see also~\cite{Blecher89} for a slightly different formulation). The corresponding norm on bilinear forms $M_n\times M_n\to \C$ is called the ``tracially completely bounded norm'' and denoted $\|\cdot \|_{tcb}$. One can check that, for Hermitian $M$, $\|M\|_{tcb} = \biasme(M)$.} This is most easily seen from Proposition 13.9 of~\cite{PisierGT}. It remains to use the equivalence shown in Proposition~18.2 of~\cite{PisierGT} between its items (i) and (ii) to conclude that our theorem is equivalent to (the finite dimensional case of) the main theorem of~\cite{HM08}, Theorem~1.1. (We remark that Pisier states in Theorem 18.1 of~\cite{PisierGT} a corollary of Theorem~1.1 of~\cite{HM08} which, if used in combination with his Proposition 18.2 implies  Theorem~\ref{thm:os-gt} with the worse constant of $4$.)